\pdfoutput=1
\documentclass[journal]{IEEEtran}
\usepackage{xr-hyper}                   
\usepackage[utf8]{inputenc}
\usepackage{amsmath}
\usepackage{amsfonts}
\usepackage{bbding}
\usepackage{amssymb}
\usepackage{array}
\usepackage{multirow}
\usepackage{graphicx}
\usepackage[caption=false,font=footnotesize]{subfig}
\usepackage{algorithm}
\usepackage{algorithmic}
\usepackage[hidelinks]{hyperref}
\newcommand\algorithmicprocedure{\textbf{procedure}}
\newcommand{\algorithmicendprocedure}{\algorithmicend\ \algorithmicprocedure}
\makeatletter
\newcommand\PROCEDURE[3][default]{%
  \ALC@it
  \algorithmicprocedure\ \textsc{#2}(#3)%
  \ALC@com{#1}%
  \begin{ALC@prc}%
}
\newcommand\ENDPROCEDURE{%
  \end{ALC@prc}%
  \ifthenelse{\boolean{ALC@noend}}{}{%
    \ALC@it\algorithmicendprocedure
  }%
}
\newenvironment{ALC@prc}{\begin{ALC@g}}{\end{ALC@g}}
\makeatother

\usepackage{xcolor}
\definecolor{blue}{rgb}{0,0,0} 
\usepackage{cite}

\allowdisplaybreaks

\usepackage{amsthm}
\usepackage{stfloats}
\newtheorem{defi}{Definition}
\newtheorem{theorem}{Theorem}
\newtheorem{lemma}{Lemma}

\newtheorem{rem}{Remark}

\newtheorem{prop}{Proposition}
\newtheorem{cor}{Corollary}
\newcommand{\blue}[1]{\textcolor{blue}{#1}}

\DeclareMathOperator*{\argmin}{\arg\!\min}
\DeclareMathOperator*{\argmax}{\arg\!\max}

\begin{document}
\title{\fontsize{21}{28}\selectfont Joint Laser Inter-Satellite Link Matching and Traffic Flow Routing in LEO Mega-Constellations via Lagrangian Duality}

\author{
  \IEEEauthorblockN{$^\dagger$Zhouyou Gu, $^\ddagger$Jinho Choi, $^\dagger$Jihong Park}\\
\thanks{
$^\dagger$Z. Gu and J. Park are with the Information Systems Technology and Design Pillar, Singapore University of Technology and Design, Singapore 487372 (email: \{zhouyou\_gu, jihong\_park\}@sutd.edu.sg).
}
\thanks{
$^\ddagger$J. Choi is with the School of Electrical and Mechanical Engineering,
the University of Adelaide, Adelaide, SA 5005, Australia
(email: \{jinho.choi\}@adelaide.edu.au).
}
\thanks{Source codes are available at \url{github.com/zhouyou-gu/leo-sat-flow}. Corresponding authors: J. Park and Z. Gu.}
\vspace{-0.75cm}
}

\maketitle

\begin{abstract}
Low Earth orbit (LEO) mega-constellations greatly extend the coverage and resilience of future wireless systems. Within the mega-constellations, laser inter-satellite links (LISLs) enable high-capacity, long-range connectivity. Existing LISL schemes often overlook mechanical limitations of laser communication terminals (LCTs) and non-uniform global traffic profiles caused by uneven user and gateway distributions, leading to suboptimal throughput and underused LCTs/LISLs -- especially when each satellite carries only a few LCTs. \blue{This paper investigates the joint optimization of LCT connections and traffic routing to maximize the constellation throughput, considering the realistic LCT mechanics and a snapshot-level spatial traffic profile.} \blue{We show that the resulting per-snapshot formulation is a mixed-integer program coupling LCT connections with flow-rate variables under link capacity constraints, and the associated fixed-snapshot abstract decision problem is NP-hard.} Due to its intractability, we resort to relaxing the coupling constraints via Lagrangian duality, decomposing the problem into a weighted graph-matching for LCT connections, weighted shortest-path routing tasks, and a linear program for rate allocation. Here, Lagrange multipliers reflect congestion weights between satellites, jointly guiding the matching, routing, and rate allocation. Subgradient descent optimizes the multipliers, with proven convergence. Simulations using real-world constellation and terrestrial data show that our methods substantially improve network throughput by up to $35\%$--$145\%$ over existing non-joint approaches.
\end{abstract}

\begin{IEEEkeywords} 
Mega-constellations; Laser links; Graph theory.
\end{IEEEkeywords}

\section{Introduction}
Low Earth orbit (LEO) satellites are emerging as a cornerstone of future generations of wireless communication systems, particularly in providing ubiquitous connectivity worldwide \cite{maiolinicapez2024use}.
Owing to their reduced propagation delays and improved link budgets to terrestrial terminals, the deployment of LEO satellite constellations complements terrestrial networks and extends the network coverage by enabling direct connections to users \cite{he2024directtosmartphone}.
Commercial proprietary programs such as SpaceX’s Starlink and OneWeb have been deploying thousands of LEO satellites, forming mega-constellations \cite{wu2025enhancing}.
The integration of laser inter-satellite links (LISLs) using free-space optical (FSO) systems in the LEO mega-constellations promises to deliver high data rates and long ranges in inter-satellite communications \cite{elamassie2023freea}, connecting the traffic flows across the whole constellation.

Compared to radio frequency (RF) links, LISLs modulate data using directional beams that emit nearly all the transmitted energy into a narrow path, thereby minimizing energy spreading \cite{ntontin2025vision}. 
However, due to the narrow beam width, one laser communication terminal (LCT) can only connect to one other LCT at a time, limiting the number of LISLs that can be established \cite{wang2024free}. 
Also, the beam steering is subject to mechanical limitations such as limited steering range and vibrations, which can further restrict the connectivity of the constellation \cite{kaymak2018survey}.
Thus, how to allocate the scarce LCTs to maximize the constellation connectivity while considering their mechanical limitations is a critical challenge.

Given the limited LCT resources on each satellite, simple designs use grid patterns that link each satellite to its immediate neighbors in the same orbital plane and to satellites in adjacent planes \cite{chaudhry2021laser,chen2021analysis}.
In this approach, the offset of the connected LCTs between adjacent orbital planes can be further adjusted to optimize the constellation connectivity \cite{rao2025minimumhop,guo2024constellation}.
Alternatively, the authors in \cite{bhattacherjee2019network} propose connecting the LISLs via optimized motifs (small, repeated connection patterns).
More flexible designs use maximum weight matching (MWM) that assigns weights to all potential LISLs and connects the LCTs based on these weights \cite{leyva-mayorga2021interplane,ron2025time}. Here, weights are assigned based on a single LISL property, e.g., maximum link rate \cite{leyva-mayorga2021interplane}, or a score that considers multiple link characteristics, e.g., capacity and latency \cite{ron2025time}.
However, these designs ignore the diverse traffic profiles across the constellation due to the uneven distribution of users and gateways.
This can lead to underutilized LCTs and LISLs and suboptimal network throughput \cite{bhattacherjee2019network}. 
Incorporating traffic flow routing into the LCT connection design is essential to improve the overall throughput of the constellation.
Recent traffic-engineering approaches such as SaTE optimize the traffic-rate allocation after the topology and routes are fixed, rather than jointly optimizing the LISL matching, routing, and rate allocation \cite{wu2025sate}.

\blue{The joint link-matching and flow-routing formulation leads to a mixed-integer optimization problem \cite{ahuja1993network}, and the associated fixed-snapshot abstract decision problem is NP-hard \cite{jarry2012multiflows} (as proved later in Proposition~\ref{prop:np_hardness} and the appendix).}
One powerful technique to solve the problem efficiently is to use Lagrangian duality \cite{fisher2004lagrangian}.
It views the integer program as a core problem plus a set of complicating constraints, and by lifting these constraints to the objective function with Lagrange multipliers (or Lagrangian dual variables), the original problem is relaxed and can be solved efficiently due to fewer constraints \cite{ahuja1993network}. 
This method has been successfully applied to optimize the network coding multicast rates \cite{wang2023joint} and per-traffic-flow energy efficiency \cite{huang2025lagrangianbased} in satellite networks by lifting bottleneck link constraints and end-to-end path length constraints, respectively.
Learning-based surrogates can also be used to infer such dual variables in real time.
\blue{Nevertheless, how to design a Lagrangian dual method that considers practical LCT mechanical limitations and a snapshot-level spatial traffic profile over the globe in LEO mega-constellations is still an open problem.}

\begin{table*}[t]
  \centering
  \caption{Summary of Notations}
  \label{tab:notation}
  \vspace{-0.15cm}
  \footnotesize
  \setlength{\tabcolsep}{4pt}
  \renewcommand{\arraystretch}{1.05}
  \begin{tabular}{@{}l p{0.35\textwidth} l p{0.35\textwidth}@{}}
    \hline
    \textbf{Symbol} & \textbf{Definition} & \textbf{Symbol} & \textbf{Definition} \\
    \hline
    $t$ & System time (s) & $\mathcal{T}_0$ & Real-world reference time at $t=0$ \\

    $I$, $\mathcal{I}$ & Number of satellites; set of satellites & $\mathbf{l}_i(t)$, $\mathbf{v}_i(t)$ & Position and velocity of satellite $i$ (m, m/s) \\

    $z_{i,j}(t)$ & Distance between satellites $i$ and $j$ (m) & $\mathbf{d}_{i,j}(t)$ & Unit direction vector from $i$ to $j$ \\

    $N'$ & Number of LCTs per satellite & $N,\ \mathcal{N}$ & Total number of LCTs; set of all LCTs \\

    $i_n$ & Satellite index that LCT $n$ belongs to & $\mathbf{u}_n$ & Mounting direction of LCT $n$ (unit vector) \\

    $\theta$ & LCT field of regard (FOR) half-angle (rad) & $\nu$ & Pointing jitter angle (rad) \\

    $\sigma_{\mathrm{J}}$ & Rayleigh scale parameter of pointing jitter & $P_0$ & Optical transmit power (W) \\

    $\Phi(y,z)$ & Gaussian beam intensity at $(y,z)$ & $\Phi_0$ & Peak intensity at beam waist \\

    $W_0$ & Beam waist radius (m) & $W(z)$ & Beam radius at distance $z$ \\

    $z_{\mathrm{R}}$ & Rayleigh range & $C(y,z)$ & LISL capacity lower bound (bps) \\

    $B$ & Receiver bandwidth (Hz) & $A$ & LCT aperture area (m$^2$) \\

    $\Psi$ & Receiver responsivity (A/W) & $\sigma_{\mathrm{N}}$ & Noise amplitude (A) \\

    $U_i(t)$ & Number of Users covered by satellite $i$ & $D$ & Per-user traffic demand (Gbps) \\

    $Q$ & Gateway service capacity per satellite (Gbps) & $Q_i(t)$ & Serving rate of satellite $i$ (Gbps) \\

    $D_i(t)$ & Demand rate of satellite $i$ (Gbps) & $\mathcal{G}^{\text{LCT}}=(\mathcal{N},\mathcal{E})$ & LCT connectability graph \\

    $\hat{z}$ & Max distance for connectable LCT pair & $\mathcal{E}$ & Set of connectable LCT pairs \\

    $c_{n,m}$ & LISL establishment indicator between LCTs $n$ and $m$ & $\mathcal{C}$ & Feasible set of LCT connections $\mathbf{c}$ \\

    $M$ & Number of nearest serving satellites considered & $\mathcal{F}$ & Set of traffic source--target satellite pairs $(s,s')$ \\

    $q^{s,s'}$ & Flow rate from $s$ to $s'$ (Gbps) & $\mathbf{q},\ \mathcal{Q}$ & Flow vector and its feasible set \\

    $\mathcal{E}_{i,j}$ & Connectable LCT pairs between sats $i$ and $j$ & $\mathcal{L}$ & Neighboring satellite pairs (potential links) \\

    $\mathcal{G}^{\text{SAT}}=(\mathcal{I},\mathcal{L})$ & Satellite connectivity graph & $x^{s,s'}_{i,j}$ & Routing indicator of flow $(s,s')$ on link $(i,j)$ \\

    $\mathbf{x},\ \mathcal{X}$ & Routing decisions and feasible set & $r_{n,m}$ & Effective LISL data rate of LCT pair $\{n,m\}$ \\

    $\epsilon$ & LISL outage probability threshold & $\lambda_{i,j}$ & Lagrange multiplier for link-rate constraint on $(i,j)$ \\
    \hline
  \end{tabular}
  \vspace{-0.25cm}
\end{table*}

This paper investigates the joint optimization of LISL matching and traffic flow routing in LEO mega-constellations, where the objective is to maximize the network throughput routed through the established LISLs. We first formulate the joint problem based on the models of constellation orbits, LCT mechanics, and traffic profiles.
To tackle this mixed-integer problem, we apply the Lagrangian dual relaxation method by lifting the maximum rate limitation constraints between neighboring satellites via Lagrange multipliers.
Each Lagrange multiplier represents the congestion level of a satellite pair. A larger multiplier indicates that matching the pair's LCTs serves higher traffic demand. It also indicates that the corresponding LISLs are more congested and should be avoided in routing.
This decouples the joint problem into three Lagrange multiplier-weighted subproblems, namely (a) a weighted graph-matching problem for LCT pairs, (b) a weighted shortest-path routing problem for each traffic flow, and (c) a linear program (LP) for rate allocation.
Based on the relaxation, we jointly optimize the LISL matching and flow routing decisions by adjusting the multipliers, e.g., via subgradient descent \cite{boyd2003subgradient}.
Our contributions are as follows.
\begin{itemize}
  \item \blue{We provide a comprehensive formulation of the problem \eqref{eq:prob:constrainted_routing_and_matching:primal}, including the LCTs' steering ranges and vibrations, and the diverse traffic profile across the constellation. Moreover, we discuss the complexity of the joint problem through the induced fixed-snapshot abstract decision problem. Specifically, we define this abstract decision problem and prove that it is NP-hard, via a direct reduction from integral 2-commodity flow in symmetric digraphs \cite{jarry2012multiflows} (\textbf{Proposition~\ref{prop:np_hardness}}).}
  \item We relax the link rate limitation constraints of satellite pairs in \eqref{eq:prob:constrainted_routing_and_matching:relaxed} and  \eqref{eq:prob:constrainted_routing_and_matching:dual}, separating the original problem into subproblems (i.e., (a), (b), and (c) in \eqref{eq:prob:constrainted_routing_and_matching:lagrangian_dual}) that are solved independently at polynomial complexity.
  \item Building on this, we design a Lagrangian dual-based subgradient-descent algorithm (\textbf{DuJo}) and prove that the set of optimal Lagrange multipliers and the subgradient are bounded, guaranteeing the convergence of the algorithm to the optimal multipliers (\textbf{Theorem} \ref{theorem:convergence:classic}).
  \item With simulations based on real-world constellation and terrestrial information, we show that our method \textbf{DuJo} can improve the network throughput in Starlink-based constellations approximately by up to $145\%$ compared to a grid-aligned LISL scheme and up to $35\%$ compared to a scheme prioritizing high-capacity LISLs.
\end{itemize}

\subsection{Notations and Paper Organization} 
We denote the $i$-th element of a vector $\mathbf{x}$ by $x_i$ and the $(i,j)$-th element of a matrix $\mathbf{x}$ by $x_{i,j}$.
Table~\ref{tab:notation} summarizes the notations used in this paper.
We collect all the elements in $\mathbf{x}$ as $\{x_{i,j}\}_{\forall(i,j)}$.
The rest of this paper is organized as follows. Section~\ref{sec:system_model} presents the system model of the LEO mega-constellation, the models of LCT mechanics, and traffic profiles.
Section~\ref{sec:problem_formulation} formulates the constellation graphs and the joint optimization problem over the graphs, and Section~\ref{sec:lagrangian_dual_relaxation} presents the Lagrangian dual relaxation of the joint problem.
Section~\ref{sec:simulation} shows the simulations evaluating the proposed methods, and Section~\ref{sec:conclusion} concludes this work.

\section{System Model of Mega-Constellation}\label{sec:system_model}
This section explains the system model of the mega-constellation connected by LISLs, as illustrated in Fig.~\ref{fig:leo_system_diagram}.

\subsection{Constellation Model}

The system time is denoted by $t$ in seconds, where $t=0$ is the initial time corresponding to a real-world time $\mathcal{T}_0$. 
The spatial domain is described in Cartesian coordinates, with the origin fixed at the center of the Earth \cite{vallado2022fundamentals}.
The satellites and the Earth are assumed to move with respect to the Earth-centered inertial (ECI) frame.
In the ECI frame, the direction of the vernal equinox is taken as the reference direction along the x-axis, and the z-axis is aligned with the Earth's rotation axis. The Earth rotates about the z-axis at an angular velocity of $7.2921 \times 10^{-5}$ rad/s. The Earth's radius is set to $6.3781 \times 10^6$ m.
We consider a mega-constellation with $I$ satellites that are collectively denoted as $\mathcal{I} = \{1,\dots,I\}$.
Each satellite has a near-circular orbit described by two-line element sets (TLEs).
The position and velocity of satellite $i \in \mathcal{I}$ at time $t$ are denoted by $\mathbf{l}_i(t)$ (m) and $\mathbf{v}_i(t)$ (m/s) in the ECI frame, respectively, and are determined by the TLEs.
The distance and direction vector from satellite $i$ to $j$ are denoted by $z_{i,j}(t)$ and $\mathbf{d}_{i,j}(t)$, respectively, which are given by
\begin{equation}
  \begin{aligned}
    z_{i,j}(t) = \|\mathbf{l}_i(t) - \mathbf{l}_j(t)\|\ 
    \text{and}\ 
    \mathbf{d}_{i,j}(t) = \frac{\mathbf{l}_j(t) - \mathbf{l}_i(t)}{z_{i,j}(t)}.
  \end{aligned}
\end{equation}

\subsection{Satellite Form Factor and LCT Mechanics Models}
We model each satellite as a rigid body that can rotate in space by using its attitude control system \cite{vallado2022fundamentals}.
The satellites maintain their orientations, with their body frames always pointing toward the center of the Earth, i.e., in the direction of $-\mathbf{l}_i(t)$ at time $t$.
Each satellite $i$ is equipped with $N'$ LISL communication transceivers (LCTs), so the constellation has a total of $N=N'I$ LCTs. We denote the set of all LCTs by $\mathcal{N} = \{1,\dots,N\}$.
We write the index of the satellite that LCT $n$ belongs to as $i_n$, where $i_n \in \{1,\dots,I\}$. 
LCT $n$ has a mounting direction $\mathbf{u}_n(t)$ in the ECI frame, which is a unit vector fixed on the satellite body.
\blue{These mounting directions specify only the preferred pointing directions of the LCTs. They do not restrict feasible LISLs to satellites in the same orbital plane \cite{chaudhry2021laser,leyva-mayorga2021interplane}.}
Each LCT can steer its beam to point in any direction with an angle less than $\theta$ radians from its center mounting direction, e.g., by using steering mirrors. Here, this region is referred to as its field of regard (FOR).
Due to spacecraft dynamics and satellite motion, the beam pointing direction of the LCT can deviate from the desired direction because of small unpredictable errors, referred to as pointing jitters.
The pointing jitter in radians is assumed to follow a Rayleigh distribution \cite{arnon2003effects} as
\begin{equation}\label{eq:pointing_jitter}
  \begin{aligned}
    f_{\text{jitter}}(\nu) = \frac{\nu}{\sigma^2_{\mathrm{J}}} \exp\left(-\frac{\nu^2}{2\sigma^2_{\mathrm{J}}}\right), \nu \geq 0.
  \end{aligned}
\end{equation}
\blue{Each LCT's pointing error is characterized by the marginal Rayleigh distribution above. Since the per-link capacity derivation in the next subsection depends only on each LCT's marginal distribution \cite{kaymak2018survey,yang2024analysis}, the framework does not require pointing errors to be statistically independent or identically distributed across LCTs, and per-LCT jitter scales can be accommodated without changing the matching or routing algorithm.}

\subsection{LCT Beam Model}
Laser beams are modeled as Gaussian beams \cite{saleh2019fundamentalsa}. 
Each beam is configured with power $P_0$ watts (W).
The intensity of a propagating Gaussian beam is given by
\begin{equation}
  \begin{aligned}
\Phi(y,z)=\Phi_0\left(\frac{W_0}{W(z)}\right)^2 \exp\!\left(-\frac{2y^2}{W(z)^2}\right),
  \end{aligned}
\end{equation}
where $\Phi_0$ is the intensity at the center of the beam at the beam waist, $y$ is the radial distance from the beam's center, $W(z)$ is the beam radius at distance $z$ from the beam waist, and $W_0$ is the beam waist radius.
Here, the peak intensity $\Phi_0$ and the beam radius $W(z)$ are modeled as
\begin{equation}
  \begin{aligned}
    \Phi_0 = \frac{2P_0}{\pi W_0^2}, \
    W(z)=W_0\sqrt{1+\left(\frac{z}{z_\mathrm{R}}\right)^2},
  \end{aligned}
\end{equation}
where $z_\mathrm{R}$ is the Rayleigh range of the beam.
Data over LISLs are transmitted using an intensity-modulation and direct-detection scheme, e.g., ON-OFF keying (OOK). While the closed-form expression for the channel capacity cannot be derived, lower bounds on the capacity have been studied.
Specifically, given $y$ and $z$, the channel capacity of the LISL is lower bounded as \cite{farid2007outage,lapidoth2009capacity,chaaban2022capacity}
\begin{equation}
  \begin{aligned}
    C(y,z) \approx \frac{B}{2}\log_2
    \left(1+\frac{\left( A \cdot \Phi(y,z)\cdot \Psi  \right)^2}{2\pi e \sigma^2_{\mathrm{N}}} \right),
  \end{aligned}
\end{equation}
where the unit is bits per second (bps), $\Psi$ is the responsivity of the LCT receiving optical sensor in ampere/watt (A/W), $B$ is the bandwidth of the receiver in Hertz (Hz), $A$ is the aperture area of the LCTs in m$^2$, and $\sigma_{\mathrm{N}}$ is the noise amplitude in ampere (A) \cite{maxim_spf_transimpedance}.
For sufficiently small pointing jitter angles, the radial distance $y$ can be approximated as $y \approx z \nu$. 
The capacity of the LISL, as a function of the pointing jitter $\nu$ and the satellite distance $z$, can be written as $C(z \nu,z)$.
\blue{
\begin{rem}
Our formulation performs per-snapshot joint design of the matching, routing, and rate allocation at a fixed time $t$. It therefore does not internally optimize acquisition, pointing, and tracking (ATP) transition overheads that arise when the matching changes across consecutive snapshots \cite{kaymak2018survey,li2011analytical,scheinfeild2000acquisition}. At snapshot intervals on the order of seconds, such ATP overheads are not negligible; for example, optical inter-satellite link experiments have reported initial acquisition times on the order of tens of seconds, with shorter times after optimization \cite{smutny2009gbps}. We evaluate the impact of non-zero ATP time in Section~\ref{subsubsec:atp_transition_evaluation}. A fully ATP-aware optimizer can be cast as an MDP on the matching state in which our current three-way Lagrangian decomposition is preserved at each snapshot; this is consistent with recent time-dependent LEO topology optimization work that explicitly accounts for link churn over time \cite{ron2025time}, and details are provided in the appendix as a future extension.
\end{rem}
}
\blue{
\begin{rem}
From a system-operation perspective, the optimization is intended to be carried out centrally by the ground segment or network-control segment rather than independently onboard the satellites \cite{3gpp38821,guidotti2024role}. For each snapshot, the controller uses predicted ephemerides, gateway availability, and traffic estimates to compute the matching, routing, and rate-allocation decisions, and then disseminates the resulting schedule to the satellites for execution. In real deployments, this workflow would also face control-plane latency, imperfect traffic or state estimates, and robustness challenges under snapshot-to-snapshot disturbances or link interruptions \cite{guidotti2024role,wu2025enhancing}. Hierarchical or distributed implementations, together with robust or online re-optimization under delayed or imperfect state information, are therefore important future directions.
\end{rem}
}


\begin{figure}[t]
  \centering
  \includegraphics[scale=0.9]{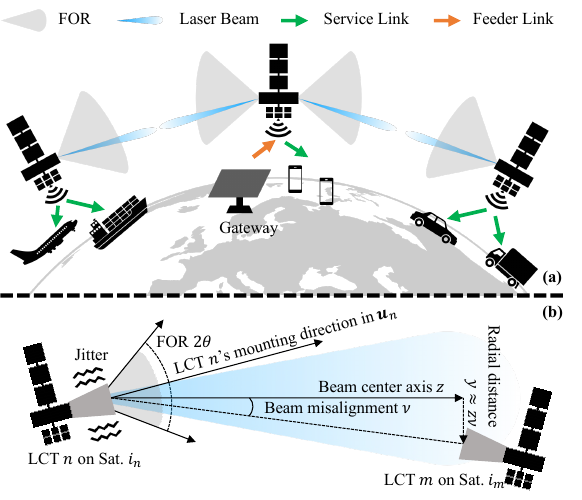}
  \vspace{-0.2cm}
  \caption{Illustration of (a) the LEO satellite mega-constellation and (b) the LCT and beam model between two LCTs.}
  \label{fig:leo_system_diagram}
  \vspace{-0.25cm}
\end{figure}

\subsection{Traffic Profile Model}\label{subsec:traffic_profile_model}

As the focus of this paper is on LCT management in space, we consider a simplified scenario where the constellation serves traffic from gateways to users under its coverage over the globe.
\blue{Accordingly, the present traffic model is intended to capture snapshot-level spatial heterogeneity through the uneven distributions of population and gateways, rather than the full demand-forecasting problem.}

It is assumed that each user is downloading data from the Internet, e.g., for video streaming via the service links. We assume that each satellite $i$ has a random number of users, denoted by $U_i(t)$, within its coverage area at time $t$, and each user has a traffic demand of $D$ Gbps.
We assume that there are multiple ground gateways scattered across the Earth, and each satellite can serve up to $Q$ Gbps of traffic if a ground gateway lies within its coverage area. 
\blue{Here, $Q$ denotes the gateway service capacity available between that satellite and the terrestrial Internet. It does not include LISL transit capacity, which is constrained separately by the aggregate capacities of the established optical links.}
Note that with a ground gateway connection, a satellite can not only serve the traffic demands of the users under its coverage but also serve traffic from other satellites in the constellation when they are connected by LISLs.
\blue{Each satellite first uses this gateway service capacity to serve the traffic demand of its own users under its coverage, i.e., $U_i(t) D$ Gbps for satellite $i$, and then uses the remaining gateway capacity to serve the traffic demand of other satellites.}
Specifically, by subtracting the traffic demand of the users under its own coverage, each satellite $i$ can serve the traffic demand of other satellites up to
\begin{equation}
  \begin{aligned}
    Q_i(t) = \begin{cases}
      (Q - U_i(t) D)_{+},\ &\text{if satellite $i$ has a gateway},\\
      0,\ &\text{otherwise},
    \end{cases}    
  \end{aligned}
\end{equation}
which is referred to as the traffic serving rate of satellite $i$ for the constellation. Here, $(x)_{+} = \max\{x,0\}$. 
Similarly, the traffic demand of satellite $i$ is the remaining part of the traffic demand of the users under its coverage that is not served by its own gateway connection (if it has one) and is denoted by
\begin{equation}
  \begin{aligned}
    D_i(t) = \begin{cases}
      (U_i(t) D-Q)_{+},\ &\text{if satellite $i$ has a gateway},\\
      U_i (t)D,\ &\text{otherwise},
    \end{cases}    
  \end{aligned}
\end{equation}
which is referred to as the traffic demand rate of satellite $i$ to the constellation.
Here, positive traffic demand of a satellite needs to be served by other satellites with a ground gateway connection via the LISLs.
\blue{This simplified model does not incorporate temporal demand variation or socioeconomic modifiers such as purchasing capability. Richer broadband-demand models along these dimensions are important future work \cite{zamacola2025profiling,qin2024satformer}. Gateway-aware traffic estimation is also relevant to this broader modeling direction \cite{guo2021gateway}.}
Satellites with positive serving rates or positive demand rates are referred to as serving and demanding satellites, respectively.

\section{Problem Formulation}\label{sec:problem_formulation}
We consider the constellation at a given time $t$, and we omit $t$ for simplicity in the rest of the paper. 
The task is to maximize the network throughput in the constellation by optimizing 1) the connecting pairs of LCTs, 2) the traffic flow allocation between serving and demanding satellites, and 3) the routing of the traffic flows through the established LISLs.
To construct the problem, we first define three independent feasible domains constraining the above groups of decision variables, and then formulate the maximum link rate limitation constraints that couple these decision variables.

\subsection{\texorpdfstring{\blue{LCT Visibility and Connection Constraints}}{LCT Visibility and Connection Constraints}}\label{subsec:lct_visibility_connection}
The LCTs $\mathcal{N}$ and all connectable LCT pairs $\mathcal{E}$ are modeled as a graph $\mathcal{G}^\text{LCT} = (\mathcal{N},\mathcal{E})$, where the vertices are the LCTs and the edges are connectable LCT pairs.
\blue{The visibility condition for a potential LISL is determined by the maximum beam transmission distance $\hat{z}$ and the mutual field of regard test at the two LCTs.}
\blue{Accordingly, the set $\mathcal{E}$ of all connectable LCT pairs is defined by the following visibility condition.}
\begin{equation}\label{eq:lct_visibility_condition}
  \begin{aligned}
    &\mathcal{E} = \big\{\{n,m\}\ \big| \ i_n \neq i_m, z_{i_n,i_m} \leq \hat{z}, \\
    &\qquad\qquad \mathbf{d}_{i_n,i_m} \cdot \mathbf{u}_n > \cos \theta,\mathbf{d}_{i_m,i_n} \cdot \mathbf{u}_m > \cos \theta \big\}.
  \end{aligned}
\end{equation}
\blue{Therefore, the model does not impose a same-plane restriction. Both intra-plane and inter-plane LISLs can be feasible whenever they satisfy this visibility condition.}
\blue{Note that this visibility model is simplified. We do not model atmospheric effects here because the present LISL formulation focuses on space-to-space propagation. We also do not impose an explicit line-of-sight or Earth-occlusion constraint, so the current connectability approximation is driven only by the range bound and the mutual field of regard test \cite{vallado2022fundamentals,leyva-mayorga2021interplane,yang2024analysis}.}
Denote by $c_{n,m}$ the binary indicator of whether the LISL from LCT $n$ to $m$ is established. If the LISL is established from LCT $n$ to $m$, it is also established from LCT $m$ to $n$, which can be mathematically expressed as
\begin{equation}\label{eq:const:link:binary_bidirectional_connection}
  \begin{aligned}
    c_{n,m}\in \{0,1\},\ c_{n,m} = c_{m,n} ,\ \forall n\neq m.
  \end{aligned}
\end{equation}
Each LCT can connect to at most one other LCT at a time, i.e.,
\begin{equation}\label{eq:const:link:one_connection}
  \begin{aligned}
    \sum_{m\in\mathcal{N}} c_{n,m} \leq 1,\ \forall n\in\mathcal{N}.
  \end{aligned}
\end{equation}
By collecting indicators $c_{n,m}$ as a vector  $\mathbf{c} = \{c_{n,m}\}_{(n,m)\in\mathcal{E}}$, the feasible domain of LCT connections is defined as
\begin{equation}
  \begin{aligned}
    \mathcal{C} = \left\{ \mathbf{c} \ \big|\   \eqref{eq:const:link:binary_bidirectional_connection},\eqref{eq:const:link:one_connection}\right\}.
  \end{aligned}
\end{equation}

\subsection{Traffic Serving and Demand Rates Constraints}
When there is no restriction, a demanding satellite $s'$ can be served by any satellite with a ground gateway connection. This will lead to a significant number of traffic source-target pairs, increasing the problem complexity.
To reduce the number of traffic source-target satellite pairs, we assume that a demanding satellite $s'$ can only be served by its $M$ nearest satellites via LISLs.
We denote the $M$ nearest satellites that can serve traffic for satellite $s'$ by $\mathrm{TopM}_{s:Q_s > 0} -z_{s,s'}$, where $z_{s,s'}$ is the distance from satellite $s$ to $s'$.
We collect all the traffic source and target satellite pairs in a set $\mathcal{F}$ as
\begin{equation}
  \begin{aligned}
    \mathcal{F} = \{(s,s')\big|D_{s'} > 0; s \in \mathrm{TopM}_{s:Q_s > 0} -z_{s,s'} \}.
  \end{aligned}
\end{equation}
We denote the amount of traffic demand of satellite $s'$ served by satellite $s$ as $q^{s,s'}$.
For all source satellites $s$ and target satellites $s'$, the traffic flow rate $q^{s,s'}$ is bounded by the traffic serving and demand rates of the satellites, i.e.,
\begin{equation}\label{eq:const:traffic:flow_rate}
  \begin{aligned}
    \sum_{s':(s,s')\in\mathcal{F}}  q^{s,s'} \leq Q_s,\forall s\in\mathcal{I};\sum_{s:(s,s')\in\mathcal{F}}  q^{s,s'} \leq D_{s'}, \forall s'\in\mathcal{I}.
  \end{aligned}
\end{equation}
We collect all traffic flow rates as $\mathbf{q} = \{q^{s,s'}\}_{(s,s')\in\mathcal{F}}$. The set of all feasible non-negative traffic flow rates is given by
\begin{equation}
  \begin{aligned}
    \mathcal{Q} = \{\mathbf{q} \ \big |\  \eqref{eq:const:traffic:flow_rate};\ q^{s,s'} \geq 0,\ \forall (s,s')\in\mathcal{F}\}.
  \end{aligned}
\end{equation}

\subsection{Traffic Routing Constraints}
The traffic will be routed among neighboring satellites.
We refer to two satellites as neighbors if they have connectable LCTs. The set of all neighboring satellite pairs $i$ and $j$ is denoted by
\begin{equation}
  \begin{aligned}
    \mathcal{L} = \{(i,j)\ \big|\  \mathcal{E}_{i,j}\neq \emptyset \},
  \end{aligned}
\end{equation}
where $\mathcal{E}_{i,j}$ collects the connectable LCT pairs between $i$ and $j$ as
\begin{equation}
  \begin{aligned}
    \mathcal{E}_{i,j} = \big\{\{n,m\}\big| i_n = i, i_m = j,\{n,m\}\in\mathcal{E}\big\}, \forall \{i, j\} \in \mathcal{I}.
  \end{aligned}
\end{equation}
Denote by $x^{s,s'}_{i,j}$ the binary indicator of whether the link from satellite $i$ to $j$ is included in the flow path from satellite $s$ to $s'$. 
The routing path decisions are subject to the flow conservation constraints on the graph $\mathcal{G}^\text{SAT} = (\mathcal{I},\mathcal{L})$ as
\begin{equation}\label{eq:const:path:flow_constraints}
\begin{aligned}
  \sum_{j:(s,j)\in \mathcal{L}} x^{s,s'}_{s,j} - \sum_{j:(j,s)\in \mathcal{L}} x^{s,s'}_{j,s} & = 1, \quad \text{(source flow)}
  \\
  \sum_{j:(s',j)\in \mathcal{L}} x^{s,s'}_{s',j} - \sum_{j:(j,s')\in \mathcal{L}} x^{s,s'}_{j,s'} & = -1, \quad \text{(target flow)} 
  \\
  \sum_{j:(i,j)\in \mathcal{L}} x^{s,s'}_{i,j} - \sum_{j:(j,i)\in \mathcal{L}} x^{s,s'}_{j,i} & = 0, \quad \forall i \in \mathcal{I} \setminus \{s,s'\}.
\end{aligned}
\end{equation}
By denoting $\mathbf{x}^{s,s'} = \{x^{s,s'}_{i,j}\}_{(i,j)\in\mathcal{L}}$ as the routing decisions from satellite $s$ to $s'$, its feasible domain is defined as
\begin{equation}\label{eq:const:path:feasible_paths}
  \begin{aligned}
    \mathcal{X}^{s,s'} = \{\mathbf{x}^{s,s'} \ \big |  \eqref{eq:const:path:flow_constraints};\ x^{s,s'}_{i,j}\in\{0,1\},\forall (i,j)\in\mathcal{L}\}.
  \end{aligned}
\end{equation}
In addition, let $\mathbf{x} = \{\mathbf{x}^{s,s'}\}_{(s,s')\in\mathcal{F}}$ be routing decisions of all flows, and its feasible domain is denoted as
\begin{equation}
  \begin{aligned}
    \mathcal{X} = \{\mathbf{x} \ \big | \ \mathbf{x}^{s,s'} \in \mathcal{X}^{s,s'},\forall (s,s')\in\mathcal{F}\}.
  \end{aligned}
\end{equation}

\subsection{LISL Maximum Data Rate Constraints}
Due to the pointing jitter, the capacity of the LISL will fluctuate as the receiving LCT's radial distance to the center of the beam changes.   
The maximum data rate on a given LISL is determined by the pointing jitter and the distance between the two LCTs $n$ and $m$ as \cite{farid2007outage,lapidoth2009capacity,chaaban2022capacity}
\begin{equation}\label{eq:rate_configuration}
  \begin{aligned}
    r_{n,m} = \max_{C'} (1-\Pr\{C(z_{i_n,i_m}\nu,z_{i_n,i_m}) < C'\}) C',
  \end{aligned}
\end{equation}
where $1-\Pr\{C(z\nu,z) < C'\}$ is the outage probability that the link capacity falls below $C'$, and 
$r_{n,m}$ is the effective data rate of the LISL averaged over time, considering the outages.
Since the satellite distance $z$ is approximately constant over a short time, the outage probability can be written in terms of the jitter as $\Pr\{C(z\nu,z) < C'\} = \Pr\{\nu > \hat{\nu}\}$, where $C' = C(z\hat{\nu},z)$.
Here, $\hat{\nu}$ is the threshold of the pointing jitter above which the link capacity is less than $C'$. By requiring the outage probability to be less than a threshold $\epsilon$, i.e., $\Pr\{\nu > \hat{\nu}\} \leq \epsilon$, the rate optimization in \eqref{eq:rate_configuration} becomes finding a misalignment threshold $\hat{\nu}$, where $\Pr\{\nu > \hat{\nu}\}$ is less than $\epsilon$.
Since the misalignment is Rayleigh distributed in \eqref{eq:pointing_jitter}, the threshold is given by $\hat{\nu} = \sigma_{\mathrm{J}} \sqrt{-2\ln(\epsilon)}$.
Substituting $\hat{\nu}$, the maximum link data rate between the LCTs $n$ and $m$ can be approximated as
\begin{equation}
  \begin{aligned}
    r_{n,m} \approx (1-\epsilon) C(z_{i_n,i_m}\sigma_{\mathrm{J}} \sqrt{-2\ln(\epsilon)},z_{i_n,i_m}), \ \forall \{n,m\} \in \mathcal{E}.
  \end{aligned}
\end{equation}
The link-rate constraint is that the sum of all traffic flows routed from satellite $i$ to $j$ must not exceed the aggregate rates of the established LISLs between $i$ and $j$, i.e.,
\begin{equation}\label{eq:const:path:flow_rate}
  \begin{aligned}
    \sum_{(s,s')\in\mathcal{F}} q^{s,s'} x^{s,s'}_{i,j} \leq \sum_{\{n,m\}\in\mathcal{E}_{i,j}} r_{n,m} c_{n,m},\ \forall (i,j)\in\mathcal{L}.
  \end{aligned} 
\end{equation}

\subsection{Joint Link Matching and Flow Routing Problem}
Based on the above definitions of the feasible domains of the decision variables and the link-rate constraint \eqref{eq:const:path:flow_rate}, we can formulate the joint link matching and flow routing problem to maximize the network throughput, $\sum_{(s,s')\in\mathcal{F}} q^{s,s'}$, as
\begin{equation}\label{eq:prob:constrainted_routing_and_matching:primal}
  \begin{aligned}
  \min_{
    \substack{
      \mathbf{c}\in\mathcal{C};\mathbf{q}\in\mathcal{Q};\mathbf{x}\in\mathcal{X}
    }
  }
  -\sum_{(s,s')\in\mathcal{F}} q^{s,s'}, \ \text{s.t.}\ \eqref{eq:const:path:flow_rate},
  \end{aligned}
\tag{\bf{P1}}
\end{equation}
where $\mathcal{C}$, $\mathcal{Q}$, and $\mathcal{X}$ are the feasible domains for the LCT connections, traffic flow rates, and routing paths, respectively.
Here, we write the problem in the standard minimization form, i.e., to maximize $\sum_{(s,s')} q^{s,s'}$, we minimize its negative.
\blue{Problem~\eqref{eq:prob:constrainted_routing_and_matching:primal} is the physical mixed-integer formulation built from the geometric/orbital system model. To analyse its hardness, we use the standard threshold decision version of this optimization problem, since any exact solver for the optimization problem could answer the threshold question by comparing the optimum throughput with $\zeta$. For a fixed snapshot and a target throughput $\zeta$, we keep the same feasible set as~\eqref{eq:prob:constrainted_routing_and_matching:primal}, namely $\mathbf{c}\in\mathcal{C}$, $\mathbf{x}\in\mathcal{X}$, $\mathbf{q}\in\mathcal{Q}$, and the link-rate constraint \eqref{eq:const:path:flow_rate}. The decision question asks whether this feasible set contains any solution whose served throughput is at least $\zeta$. This fixed-snapshot threshold decision problem is the abstract snapshot joint matching and routing (AS-JMR) problem defined as follows.}

{\color{blue}
\begin{defi}[AS-JMR]
For fixed snapshot inputs $(\mathcal{I},\mathcal{N},\{i_n\}_{n\in\mathcal{N}},\mathcal{E},\mathcal{F},\{r_{n,m}\}_{\{n,m\}\in\mathcal{E}},\{Q_i\}_{i\in\mathcal{I}},\{D_i\}_{i\in\mathcal{I}},\zeta)$, with the same meanings as in~\eqref{eq:prob:constrainted_routing_and_matching:primal}, the AS-JMR decision problem asks whether there exist $\mathbf{c}\in\mathcal{C}$, $\mathbf{q}\in\mathcal{Q}$, and $\mathbf{x}\in\mathcal{X}$ satisfying \eqref{eq:const:path:flow_rate} and $\sum_{(s,s')\in\mathcal{F}}q^{s,s'}\geq\zeta$.
\end{defi}

\begin{prop}\label{prop:np_hardness}
The decision problem AS-JMR is NP-hard.
\end{prop}

\begin{proof}
  The proof is listed in the appendix.
\end{proof}

\begin{rem}
\blue{Proposition~\ref{prop:np_hardness} establishes NP-hardness for the fixed-snapshot AS-JMR decision problem with combinatorial snapshot inputs. It does not establish NP-hardness for every geometric/orbital instance of~\eqref{eq:prob:constrainted_routing_and_matching:primal}, whose connectable pairs, source-target pairs, and link rates are induced by orbital geometry, field-of-regard constraints, distance thresholds, and traffic-pair construction rules. Extending the hardness result to the full geometric/orbital formulation would require a realization argument that maps AS-JMR inputs to feasible geometric/orbital instances.
In extreme cases, the geometric/orbital constraints could yield trivial instances of~\eqref{eq:prob:constrainted_routing_and_matching:primal} that are not NP-hard, e.g., when there are no connectable LCT pairs or no source-target pairs.}
\end{rem}
\blue{Nevertheless, the NP-hardness of AS-JMR implies that the difficulty of the joint link matching and flow routing problem is inherent in the combinatorial nature of the problem, which motivates us to design efficient algorithms based on Lagrangian dual relaxation to solve the problem in the following section.}
}
\section{Joint Link Matching and Flow Routing via Lagrangian Duality}\label{sec:lagrangian_dual_relaxation}
From the structure of \eqref{eq:prob:constrainted_routing_and_matching:primal}, we observe that the maximum link-rate constraints \eqref{eq:const:path:flow_rate} couple the decisions on the LCT connections, traffic flow rates, and routing paths, while the remaining constraints are independent and define the separate feasible domains $\mathcal{C}$, $\mathcal{Q}$, and $\mathcal{X}$.
Therefore, lifting \eqref{eq:const:path:flow_rate} to the objective function with the Lagrange multipliers decouples the constraints on the LCT connections, the flow rate allocation, and the routing.
These decisions affect each other via the Lagrange multipliers that penalize links with violations of their maximum link-rate constraints. 
A larger violation on a link implies that it will have a high traffic load, which further indicates that the link should be prioritized for connection but avoided in traffic-flow routing paths.
By adjusting the Lagrange multipliers, e.g., via subgradient descent, we can jointly optimize the LCT connections, flow rates, and routing paths, while making these decisions separately.

\subsection{Lagrangian Dual Relaxation}\label{subsec:lagrangian_dual_relaxation}
\begin{figure*}[t]
\begin{align}
      g(\lambda)&= \min_{
        \substack{
          \mathbf{c}\in\mathcal{C};\mathbf{q}\in\mathcal{Q};\mathbf{x}\in\mathcal{X}
        }
      }
      L(\mathbf{c}, \mathbf{q}, \mathbf{x}, \lambda) = \min_{
        \substack{
          \mathbf{c}\in\mathcal{C};\mathbf{q}\in\mathcal{Q};\mathbf{x}\in\mathcal{X}
        }
      } -\sum_{(s,s')\in\mathcal{F} } q^{s,s'} + \sum_{(i,j)}\lambda_{i,j} (\sum_{(s,s')\in \mathcal{F}} q^{s,s'} x^{s,s'}_{i,j} - \sum_{\{n,m\}\in\mathcal{E}_{i,j}} r_{n,m} c_{n,m}) \notag
      \\
      &= \min_{
        \substack{
           \mathbf{q}\in \mathcal{Q};\mathbf{x}\in\mathcal{X}
        }
      }
      \Big\{
      \sum_{(s,s')\in\mathcal{F} } q^{s,s'} (-1 + \sum_{(i,j)\in\mathcal{L}}\lambda_{i,j} x^{s,s'}_{i,j})   
      \Big\} 
      +\min_{
        \substack{
          \mathbf{c}\in\mathcal{C}
        }
      }
      \Big\{
      -\sum_{\{n,m\}\in\mathcal{E}} (\lambda_{i_n,i_m} + \lambda_{i_m,i_n})  r_{n,m} c_{n,m}
      \Big\} \notag
      \\
      &= 
      -\underbrace{\max_{
        \substack{
          \mathbf{q}\in \mathcal{Q}
        }
      }
      \Big\{
      \sum_{(s,s')\in\mathcal{F} } q^{s,s'} (1 - 
      \overbrace{\min_{
        \substack{
          \mathbf{x}^{s,s'}\in\mathcal{X}^{s,s'}
        }
      }
      \sum_{(i,j)\in\mathcal{L}}\lambda_{i,j} x^{s,s'}_{i,j}}^{
      \substack{\text{(b) Minimum cost routing}}
      })
      \Big\}}_{\substack{\text{(c) Linear weighted flow rate maximization}}}
      -
      \overbrace{\max_{
        \substack{
          \mathbf{c}\in\mathcal{C}
        }
      }
      \Big\{
      \sum_{\{n,m\}\in\mathcal{E}} (\lambda_{i_n,i_m} + \lambda_{i_m,i_n})  r_{n,m} c_{n,m}
      \Big\}
      }^{
      \substack{\text{(a) Maximum weight LISL matching}}
      }. \label{eq:prob:constrainted_routing_and_matching:lagrangian_dual}
  \end{align}
  \hrule
\end{figure*}

Specifically, the Lagrangian-augmented objective of \eqref{eq:prob:constrainted_routing_and_matching:primal} after relaxing the maximum link rate constraints \eqref{eq:const:path:flow_rate} is
\begin{equation}\label{eq:prob:constrainted_routing_and_matching:lagrangian}
  \begin{aligned}
      &L(\mathbf{c}, \mathbf{q}, \mathbf{x}, \lambda) = -\sum_{(s,s')\in\mathcal{F} } q^{s,s'} \\
      &\qquad + \sum_{(i,j)}\lambda_{i,j} \Big(\sum_{(s,s')\in \mathcal{F}} q^{s,s'} x^{s,s'}_{i,j} - \sum_{\{n,m\}\in\mathcal{E}_{i,j}} r_{n,m} c_{n,m} \Big) ,
  \end{aligned}
\end{equation}
where $\lambda=\{\lambda_{i,j}\}_{(i,j)\in\mathcal{L}}$ is the set of Lagrange multipliers (or dual variables) for all neighboring satellite pairs $(i,j)\in\mathcal{L}$.
The relaxed version of \eqref{eq:prob:constrainted_routing_and_matching:primal} is to minimize the Lagrangian-augmented objective $L(\mathbf{c}, \mathbf{q}, \mathbf{x}, \lambda)$ over the decision variables' feasible domains as
\begin{equation}\label{eq:prob:constrainted_routing_and_matching:relaxed}
  \begin{aligned}
  \{  \hat{\mathbf{q}}(\lambda), \hat{\mathbf{x}}(\lambda), \hat{\mathbf{c}}(\lambda) \}= \argmin_{
      \substack{
        \mathbf{c}\in\mathcal{C};\mathbf{q}\in\mathcal{Q};\mathbf{x}\in\mathcal{X}
      }
    }
    L(\mathbf{c}, \mathbf{q}, \mathbf{x}, \lambda),
  \end{aligned}
  \tag{\textbf{P2}}
\end{equation}
where $\hat{\mathbf{q}}(\lambda)$, $\hat{\mathbf{x}}(\lambda)$, and $\hat{\mathbf{c}}(\lambda)$ are the optimal decisions given $\lambda$ for the relaxed problem.
The optimal value of the relaxed problem \eqref{eq:prob:constrainted_routing_and_matching:relaxed} given $\lambda$ is referred to as the dual function $g(\lambda)$, which is expressed as \eqref{eq:prob:constrainted_routing_and_matching:lagrangian_dual}.
The dual problem is to maximize the dual function $g(\lambda)$ over the Lagrange multipliers $\lambda$ as
\begin{equation}\label{eq:prob:constrainted_routing_and_matching:dual}
  \begin{aligned}
    \max_{\lambda\geq 0} g(\lambda) = \max_{\lambda\geq 0} \min_{
      \substack{
        \mathbf{c}\in\mathcal{C};\mathbf{q}\in\mathcal{Q};\mathbf{x}\in\mathcal{X}
      }
    }
    L(\mathbf{c}, \mathbf{q}, \mathbf{x}, \lambda),
  \end{aligned}
\tag{\textbf{P3}}
\end{equation}
which is a concave maximization problem whose optimal value is a lower bound for the original problem \eqref{eq:prob:constrainted_routing_and_matching:primal} \cite{fisher2004lagrangian}.

\subsection{Subgradient Descent for Dual Problem}\label{subsec:subgradient_descent}
The dual function $g(\lambda)$ is concave in $\lambda$, so the dual problem \eqref{eq:prob:constrainted_routing_and_matching:dual} is a concave maximization problem \cite{boyd2004convex} that can be solved by subgradient descent \cite{boyd2003subgradient}.
To solve \eqref{eq:prob:constrainted_routing_and_matching:dual}, the method iteratively updates the Lagrange multipliers $\lambda$ based on the supergradient of $g(\lambda)$ by solving the relaxed problem \eqref{eq:prob:constrainted_routing_and_matching:relaxed} (since $g(\lambda)$ is concave, the supergradient is used), which approximates the minimum of the original problem \eqref{eq:prob:constrainted_routing_and_matching:primal}.
For given Lagrange multipliers in the $k$-th iteration, denoted by $\lambda^{[k]}$, $k=1,\dots,K$, the supergradient of $g(\lambda^{[k]})$ is
\begin{equation}\label{eq:dual_function_subgradient}
  \begin{aligned}
   \delta(\lambda^{[k]})_{i,j}\!=\!\!\!
   \sum_{(s,s')\in\mathcal{F}}\!\!\! \hat{q}^{s,s'}(\lambda^{[k]}) \hat{x}^{s,s'}_{i,j}(\lambda^{[k]}) - \!\!\!\!\!\!\sum_{\{n,m\}\in\mathcal{E}_{i,j}} \!\!\!\!r_{n,m} \hat{c}_{n,m}(\lambda),
  \end{aligned}
\end{equation}
where $\hat{q}^{s,s'}(\lambda^{[k]})$, $\hat{x}^{s,s'}_{i,j}(\lambda^{[k]})$, and $\hat{c}_{n,m}(\lambda^{[k]})$ are the optimal decisions given $\lambda^{[k]}$ in \eqref{eq:prob:constrainted_routing_and_matching:relaxed}.
To obtain these optimal decisions, we observe that the relaxed problem \eqref{eq:prob:constrainted_routing_and_matching:relaxed} can be decomposed into three independent subproblems, which can be solved separately, as shown in \eqref{eq:prob:constrainted_routing_and_matching:lagrangian_dual}.

In detail, part (a) of \eqref{eq:prob:constrainted_routing_and_matching:lagrangian_dual} is an MWM problem on $\mathcal{G}^\text{LCT}=(\mathcal{N},\mathcal{E})$ with the edge weights $ \{\lambda^{[k]}_{i_n,i_m}r_{n,m}\}_{\{n,m\}\in\mathcal{E}}$, solved as
\begin{equation}\label{eq:routine:mwm_given_lambda}
  \begin{aligned}
    \hat{\mathbf{c}}(\lambda^{[k]}) = \mathrm{MWM}(\mathcal{N}, \mathcal{E}, \{\lambda^{[k]}_{i_n,i_m}r_{n,m}\}_{\{n,m\}\in\mathcal{E}}).
  \end{aligned}
\end{equation}
Here, the matching is heuristically approximated by greedy weight matching, which sequentially matches the LCT pairs from higher to lower weights.
The part (b) is a minimum cost path routing problem that can be efficiently solved by the shortest path first (SPF) algorithm, e.g., Dijkstra's algorithm.
Routing is performed on the graph $\mathcal{G}^\text{SAT}=(\mathcal{I}, \mathcal{L})$ with the edge weights $\lambda^{[k]}$ for all source-target satellite pairs as
\begin{equation}\label{eq:routine:spf_given_lambda}
  \begin{aligned}
    \hat{\mathbf{x}}^{s,s'}(\lambda^{[k]}) = \mathrm{SPF}(s,s',\mathcal{I}, \mathcal{L}, \lambda^{[k]}),\ \forall (s,s')\in\mathcal{F}.
  \end{aligned}
\end{equation}
After obtaining the routing costs $\sum_{(i,j)\in\mathcal{L}}\lambda^{[k]}_{i,j} \hat{x}^{s,s'}_{i,j} (\lambda^{[k]})$, the part (c) computes the rates of all source-target satellite pairs by solving the flow rate maximization (FRM) problem as
\begin{equation}\label{eq:routine:flow_rate_given_routing_cost}
  \begin{aligned}
    \hat{\mathbf{q}}(\lambda^{[k]}) = \argmax_{
        \substack{
          \mathbf{q}\in\mathcal{Q}
        }
      }
      \sum_{(s,s')\in\mathcal{F} } q^{s,s'} (1 - \sum_{(i,j)\in\mathcal{L}}\lambda^{[k]}_{i,j} \hat{x}^{s,s'}_{i,j} (\lambda^{[k]})),
  \end{aligned}
\end{equation}
where routing costs are included as the weights of the traffic flow rates in the objective. Higher routing costs imply that the flow is routed through links with higher violations, and thus that the flow rate should be reduced in the FRM.
Note that all constraints in \eqref{eq:const:path:flow_rate} as well as the objective function are linear in \eqref{eq:routine:flow_rate_given_routing_cost}, which implies that the problem is an LP that can be solved efficiently.
Finally, using the computed supergradient in \eqref{eq:dual_function_subgradient}, the Lagrange multipliers are updated as
\begin{equation}\label{eq:routine:dual_variable_update}
  \begin{aligned}
    \lambda^{[k+1]} = \big(\lambda^{[k]} + \alpha^{[k]} \delta(\lambda^{[k]})\big)_{+},
  \end{aligned}
\end{equation}
where $\alpha^{[k]}$ is the step size in the $k$-th iteration. The initial Lagrange multipliers are $\lambda^{[1]} = \mathbf{0}$.
The step size $\alpha^{[k]}$ is designed to follow a diminishing rule, i.e., $\alpha^{[k]} = \frac{\alpha_0}{k^\beta}$, where $\alpha_0$ is a positive constant and $0.5\leq\beta<1$ controls the step sizes' decay rate. \blue{Because several flows may compete for the same bottleneck LISLs, the subgradient iterates need not be monotone: the dual variables and the converted primal solutions can oscillate before settling. The diminishing step size attenuates these oscillations asymptotically, so the convergence guarantee below is stated for the best-so-far dual iterate rather than for every individual iterate.} The algorithm is listed in Algorithm~\ref{alg:dual_optimization} and is referred to as the ``DuJo'' scheme.

\begin{algorithm}[!t]
\caption{Lagrangian Dual Optimization for Joint Link Matching and Traffic Routing (\textbf{DuJo}) in Mega-Constellation}\label{alg:dual_optimization}
\begin{algorithmic}[1]
\STATE Initialize the constellation and traffic information.
\STATE Initialize Lagrange multipliers $\lambda^{[1]} = \mathbf{0}$.
\FOR{ $k = 1,2,\dots,K$}
  \STATE Solve the MWM of the constellation as \eqref{eq:routine:mwm_given_lambda}.
  \STATE Solve the SPF on the constellation as \eqref{eq:routine:spf_given_lambda}.
  \STATE Solve the FRM for all flows as \eqref{eq:routine:flow_rate_given_routing_cost}.
  \STATE Compute the supergradient as \eqref{eq:dual_function_subgradient}.
  \STATE Update the Lagrange multipliers as \eqref{eq:routine:dual_variable_update}.
\ENDFOR
\STATE Convert the optimized Lagrange multipliers as \eqref{eq:rounding:mwm}-\eqref{eq:rounding:frm}.
\end{algorithmic}
\end{algorithm}

\subsection{Convergence and Complexity of Subgradient Descent}\label{subsec:subgradient_convergence_and_complexity_analysis}
The convergence of the subgradient descent in Algorithm~\ref{alg:dual_optimization} is guaranteed by the facts that 1) the initial Lagrange multipliers are not chosen arbitrarily far from the optimum, and 2) the subgradient is numerically stable with a bounded norm, which is proved in our case as follows.
\begin{lemma}\label{lemma:optimal_dual_variable_bounded}
    The distance from the initial Lagrange multipliers to the optimal ones and the maximum norm of the dual function's supergradient are both bounded by finite constants, i.e., there exist constants $0<\eta_1<\infty$ and $0<\eta_2<\infty$ such that
    \begin{equation}
      \begin{aligned}
        \|\lambda^{[1]}-\lambda^*\|^2 \leq \eta_1; \ \max_\lambda \|\delta(\lambda)\| \leq \eta_2.
      \end{aligned}
    \end{equation} 
\end{lemma}
\begin{proof}
    The proof is listed in the appendix.
\end{proof}

Applying Lemma~\ref{lemma:optimal_dual_variable_bounded}, we show that the convergence of the subgradient descent algorithm is guaranteed as follows \cite{boyd2003subgradient}.
\begin{theorem}\label{theorem:convergence:classic}
\blue{Define the best-so-far dual iterate as $\hat{\lambda}^{[k]} \in \argmax_{\lambda^{[i]}:i=1,\dots,k} g(\lambda^{[i]})$ after $k$ iterations. The corresponding dual value converges to the optimal dual value $g(\lambda^*)$ as $k\to\infty$, i.e., $\lim_{k\to\infty} g(\hat{\lambda}^{[k]}) = g(\lambda^*)$. Specifically, the dual-value gap satisfies $g(\lambda^*) - g(\hat{\lambda}^{[k]}) = \mathcal{O}(k^{-(1-\beta)})$ when $0.5\leq\beta<1$.}
\end{theorem}
\begin{proof}
    The proof relies on the convergence of p-series in the step sizes \cite{boyd2003subgradient} and Lemma \ref{lemma:optimal_dual_variable_bounded}, as listed in the appendix.
\end{proof}
For simplicity, we return the last iterate of the Lagrange multipliers, $\lambda^{[K]}$, as the optimized Lagrange multipliers $\lambda$.

\begin{cor}\label{cor:complexity:dual_optimization}
    The complexity of the dual optimization is $\mathcal{O}\big(K(E\log E + I E\log N  + \mathrm{poly}(I))\big)$, where $E=|\mathcal{E}|$ is the number of connectable LCT pairs in the constellation.
\end{cor}
\begin{proof}
The number of source-target satellite pairs is bounded as $|\mathcal{F}| \leq I \cdot M\approx\mathcal{O}(I)$.
To approximate the MWM using the greedy weight matching in \eqref{eq:routine:mwm_given_lambda}, we need to sort the edge weights, which takes $\mathcal{O}(E\log E)$ complexity. When matching the LCT pairs, we need to iterate through all edges in the constellation and check whether the LCT pairs have been matched before, which takes $\mathcal{O}(E)$ complexity. Dijkstra's algorithm for computing the shortest path in \eqref{eq:routine:spf_given_lambda} takes $\mathcal{O}(|\mathcal{F}|\cdot E\log N )\approx\mathcal{O}(I E\log N )$ for all $|\mathcal{F}|$ source-target satellite pairs \cite{cormen2022introduction}.
The LP problem in \eqref{eq:routine:flow_rate_given_routing_cost} can be efficiently solved by the simplex method, taking polynomial time in the number of constraints in $\mathcal{Q}$, i.e., $\mathcal{O}(\mathrm{poly}(I))$ \cite{vershynin2009hirsch}.
Collecting these complexities, we have the overall complexity as the statement.
\end{proof}

\subsection{Converting Lagrange Multipliers to Matching/Routing}\label{subsec:dual_variable_guided_matching_and_routing}
For given Lagrange multipliers $\lambda$, their values indicate the weights of the maximum link rate constraints in \eqref{eq:const:path:flow_rate} for each neighboring satellite pair $(i,j)\in\mathcal{L}$.
For example, a larger $\lambda_{i,j}$ indicates that the rate constraint for the pair $(i,j)$ is harder to satisfy and that traffic will heavily load that link.
Thus, we should match the LCTs between the neighboring satellites $i$ and $j$ to allow traffic flows.
On the other hand, this also means that we should try to avoid routing traffic from $i$ to $j$ as the link is more likely to be congested. 
Based on this intuition, we can obtain a feasible solution from the optimized Lagrange multipliers $\lambda$ by sequentially 1) matching the highest-weighted LCT pairs, 2) routing the traffic flows from source to target satellites through the connected LCT pairs with the minimum cost, and 3) computing the maximum traffic flow rates based on the routing paths and the connected LCT pairs.

Specifically, first, we use the greedy weight matching to compute the MWM of the graph $\mathcal{G}^{\text{LCT}}(\mathcal{N},\mathcal{E})$ for the given Lagrange multipliers $\lambda$ (the same as \eqref{eq:routine:mwm_given_lambda}) as
\begin{equation}\label{eq:rounding:mwm}
  \begin{aligned}
    \hat{\mathbf{c}} \leftarrow \mathrm{MWM}(\mathcal{N}, \mathcal{E}, \{\lambda_{i_n,i_m} \cdot r_{n,m}\}_{(n,m)\in \mathcal{E}}).
  \end{aligned}
\end{equation}
The connected satellite pairs in the above are given by
\begin{equation}
  \begin{aligned}
    \mathcal{L}' =\{(i,j)| \exists (n,m) \in \mathcal{E}_{i,j}, c_{n,m} = 1 \}.
  \end{aligned}
\end{equation}
Based on these actual connected satellite pairs $(i,j)\in\mathcal{L}'$, we use Dijkstra's algorithm to compute the shortest path from each source satellite $s$ to each target satellite $s'$ in the graph $(\mathcal{I}, \mathcal{L}')$ with the edge weights $\{\lambda_{i,j}\}_{(i,j)\in \mathcal{L}'}$ as
\begin{equation}\label{eq:rounding:spf}
  \begin{aligned}
    \tilde{\mathbf{x}}^{s,s'} \leftarrow \mathrm{SFP}(s,s',\mathcal{I}, \mathcal{L}', \{\lambda_{i,j}\}_{(i,j)\in \mathcal{L}'}), \ \forall (s,s')\in\mathcal{F}.
  \end{aligned}
\end{equation}
Note that not all source-target satellite pairs $(s,s')\in\mathcal{F}$ have a routing path in the connected constellation, and we collect those connected source-target satellite pairs $(s,s')$ as
\begin{equation}
  \begin{aligned}
    \mathcal{F}' = \{(s,s')\in\mathcal{F} \,|\, \tilde{\mathbf{x}}^{s,s'}\ \text{is feasible in}\ \text{\eqref{eq:rounding:spf}}\},
  \end{aligned}
\end{equation}
which reduces the number of source-target satellite pairs to allocate traffic rates.
Finally, for these connected satellite pairs $(s,s')\in\mathcal{F}'$, we optimize the traffic flow rates based on the connected LCT pairs $\hat{\mathbf{c}}$ and their routing paths in $\tilde{\mathbf{x}}^{s,s'}$ as
\begin{equation}\label{eq:rounding:frm}
  \begin{aligned}
    \tilde{\mathbf{q}} \leftarrow \argmax_{
        \substack{
          \mathbf{q}\in \mathcal{Q}
        }
      }
      \sum_{(s,s')\in\mathcal{F}'}  q^{s,s'}, \ \text{s.t.}\ \eqref{eq:const:path:flow_rate}\ \text{given}\ \hat{\mathbf{c}},\ \tilde{\mathbf{x}}^{s,s'}.
  \end{aligned}
\end{equation}

Similar to Corollary~\ref{cor:complexity:dual_optimization}, the above conversion has complexity $\mathcal{O}(E\log E + I E\log N  + \mathrm{poly}(I+E))$.

\begin{figure}[!t]
  \centering
  \includegraphics[width=.95\columnwidth]{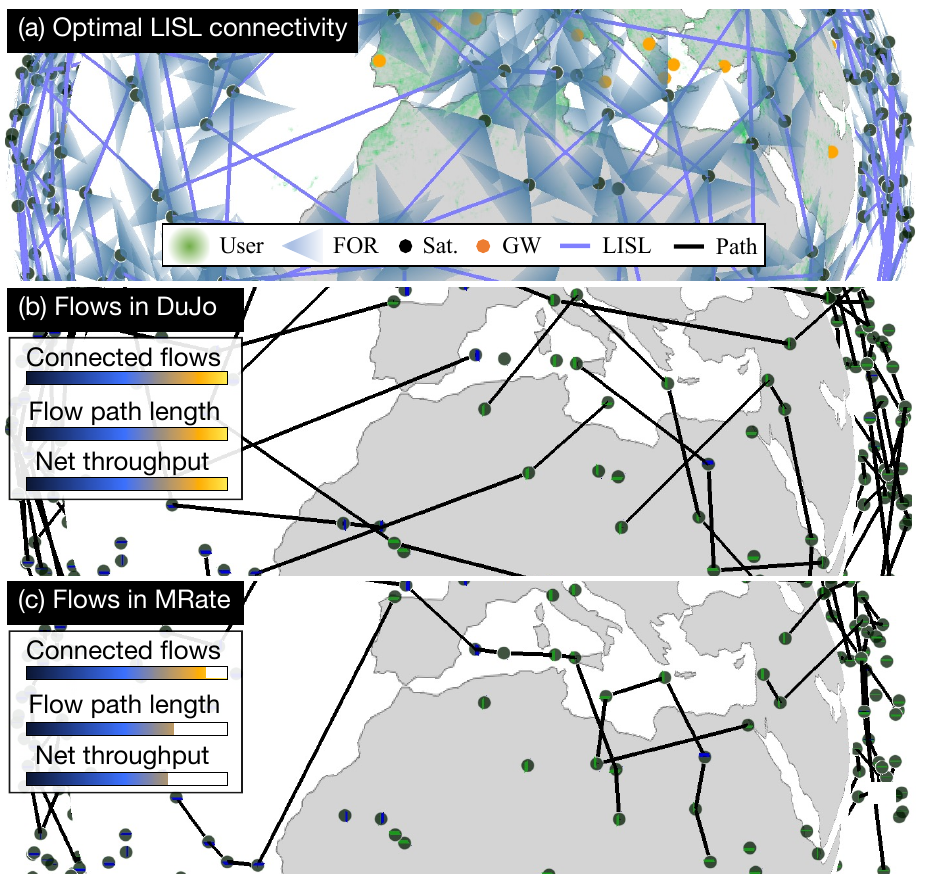}
  \caption{\blue{Comparison for the simulated Starlink-based constellation instance with $I=1000$ satellites. Panel (a) shows the constellation snapshot together with the LISLs selected by DuJo, the user distribution, and the ground-gateway locations. Panels (b) and (c) show the same geographic zoom and the same traffic-demand snapshot under DuJo and MRate, respectively, so their routed flow paths can be compared directly. The inset boxes in panels (b) and (c) report the resulting number of connected flows, flow-path length, and net throughput for the two schemes.}}
  \label{fig:constellation_simulation_1000_compare}
\end{figure}

\section{Simulation Results}\label{sec:simulation}
\subsection{Simulation Setup}
To efficiently process the graphs and simulate the constellation, we implement our algorithms and simulations in a data-oriented programming paradigm. The constellation data is organized into contiguous data arrays and is processed in just-in-time (JIT) compiled procedures using Numba for Python \cite{lam2015numba} for high computational efficiency. Linear programming is solved using the HiGHS solver \cite{huangfu2018parallelizing}. 
The simulations run on a workstation with an Intel Core Ultra 9 285K (24 cores) and 32 GB of memory.
\subsubsection{System Parameters}
We construct the constellation in simulations based on the real-world Starlink constellation two-line element (TLE) data from CelesTrak \cite{CelesTrak} at $\mathcal{T}_0$ as UTC $2025-07-16$ $16:00$. \blue{To study problem-scale sensitivity, we randomly select $I$ satellites out of the Starlink constellation for the varying-$I$ and computing-time results. We also test a structured shell-based Starlink instance with $I=1000$ satellites. Specifically, at $\mathcal{T}_0$ we select a coherent 1000-satellite block from the dominant $53.2^\circ$ shell and then keep the same satellites over snapshot offsets of $0$, $30$, $60$, $90$, $120$, and $180$ min.}
Each satellite has $N'=2$ LCTs: one LCT points in the direction of the satellite's velocity, and the other points in the opposite direction. The attitude control system maintains these orientations.
We configure the beam parameters as follows \cite{kaymak2018survey}: Each LCT has an aperture area of $A=0.01$ $\text{m}^2$ and a responsivity of $\Psi=0.5$ $\text{A/W}$. The noise current is $\sigma_{\mathrm{N}}=3\times 10^{-7}$~A in root-mean-square (RMS) value \cite{maxim_spf_transimpedance}.
The power of the beam is $P_0=20$ W and its bandwidth is $B=1$ GHz. The beam's wavelength is $1.55$ $\mu$m and the beam width is $100$ micro-radians, leading to the beam waist radius $W_0=9.87\times10^{-3}$ m and the Rayleigh range $z_\mathrm{R}=1.97\times10^{3}$ m \cite{saleh2019fundamentalsa}. 
The pointing jitter $\sigma_{\mathrm{J}}$ is set to $10$ micro-radians and the relaxed link outage probability $\epsilon$ is set to $0.001$.
The field of regard half-angle of each LCT is set to $\theta = 60$ degrees from its mounting direction, and the maximum distance for two LCTs to establish a link is $\hat{z}=3000$ km.
\blue{Under this setting, candidate LISLs are not restricted to intra-plane neighbors. Any satellite pair, whether intra-plane or inter-plane, is included if it satisfies \eqref{eq:lct_visibility_condition}.}
The number $U_i$ of active users served by satellite $i$ is determined from the real-world population data \cite{schiavina2023ghspop} within an approximately $200$ km coverage radius. We assume that $0.01\%$ of that population uses the network at the given time, so $U_i$ is Poisson distributed with a mean equal to the covered population. \blue{For the load-stress evaluation, we vary this active-user percentage from $0.0005\%$ to $0.1\%$ while keeping the user-demand model otherwise unchanged.} Each user demands $D=0.1$ Gbps. We randomly select $100$ ground-gateway locations around the globe from SatNOGS \cite{satnogs}. We assume that each satellite can serve up to $Q=20$ Gbps when a ground gateway lies within its coverage area. \blue{Thus, $Q=20$ Gbps is the per-satellite gateway service capacity before subtracting the traffic of the users under that satellite's own coverage, rather than a limit that already includes LISL transit traffic.} If a satellite has positive traffic demand, it routes that demand to the nearest $M=5$ satellites that have gateway connections, thereby defining the source-target pairs in $\mathcal{F}$.
Fig.~\ref{fig:constellation_simulation_1000_compare} shows the constellation with $I=1000$ together with the uneven distributions of population density and ground-gateway locations, which emphasize the need to account for these factors in our study.
It also compares the routed flow paths and network throughput over the constellation between our DuJo scheme and a non-joint scheme (MRate, as introduced below).
In the remaining simulations, we consider the constellation with $I=1000$ without explicit notation.

\begin{figure}[!t]
  \centering
  \includegraphics[scale=0.85]{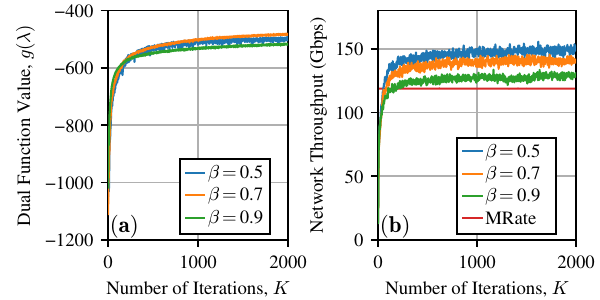}
  \caption{Convergence of DuJo. Panel (a) shows the subgradient-descent convergence, and panel (b) shows the throughput obtained after converting the Lagrange multipliers.}  \label{fig:plot_test_dual_optimization_starlink_constellation_varying_beta_d_o_and_p_o_horizontal}
\end{figure}
\subsubsection{\blue{Baseline Methods}} 
We compare our duality-based joint method, denoted in the legend as ``\textbf{DuJo},'' with the following baseline methods for LISL matching, traffic routing, or learned decision inference:
\begin{itemize}
    \item \textbf{+Grid:} A grid-based link matching \cite{bhattacherjee2019network} prioritizes the LCT pairs that are more aligned in their pointing directions, e.g., pairs with higher $\mathbf{d}_{i_n,i_m} \cdot \mathbf{u}_{m}+\mathbf{d}_{i_m,i_n} \cdot \mathbf{u}_{n}$.
    \item \blue{\textbf{MRate:} Maximum link rate matching sets the weights of LCT pairs to their maximum transmission rates, $r_{n,m}$, to emphasize link-capacity-oriented topology design \cite{guo2024constellation}. Note that in both MRate and +Grid, traffic flows are then routed by open shortest path first (OSPF) \cite{moy1998ospf} using link weights that are the reciprocals of the aggregate satellite-pair rates $1/\sum_{\{n,m\}\in \mathcal{E}_{i,j}} c_{n,m}r_{n,m}$, followed by LP-based flow rate allocation.}
    \item \blue{\textbf{SaTE:} Following the non-joint traffic-engineering setting of SaTE \cite{wu2025sate}, this baseline first fixes the LISL matching by maximum link rate matching and then routes the flows by shortest paths under unit hop cost. Given that topology and routing, the traffic flow rates are computed optimally by LP to maximize the network throughput. Thus, only the traffic rate allocation is optimized, while the link matching and traffic routing are not jointly optimized.}
    \item \blue{\textbf{DRL:} A trained policy-gradient baseline \cite{sutton2000policy} uses the constellation graph as its state, where each satellite node carries its gateway service capacity and traffic demand and each connectable satellite pair carries its candidate LISL capacity. Its action is to infer the price graph over the connectable satellite pairs, and its reward is derived from the resulting traffic-rate allocation after converting those inferred prices to matching, routing, and rate allocation.}
    \item \blue{\textbf{Rand}: Random matching first assigns random weights to the candidate links in the same feasible connectability set $\mathcal{E}$ and then matches the links from the largest weights. The traffic flows are then routed by OSPF on the resulting topology, where the rates of the paths are allocated by LP.}
\end{itemize}
\blue{In the structured-shell comparison, we focus the main benchmark on DuJo, DRL, SaTE, MRate, and +Grid.}

\subsection{Numerical Results}
\subsubsection{Convergence of Dual Function}
The convergence of the dual function $g(\lambda)$ is illustrated in Fig.~\ref{fig:plot_test_dual_optimization_starlink_constellation_varying_beta_d_o_and_p_o_horizontal}a for different $\beta$ values. 
\blue{As expected for subgradient methods, the per-iteration dual values are not necessarily monotone and may fluctuate when multiple flows compete for the same bottleneck LISLs. Nevertheless, the diminishing step size progressively damps these fluctuations, while Theorem \ref{theorem:convergence:classic} concerns the best-so-far dual value rather than every individual iterate.}
The performance of the converted LISL matching and flow routing decisions from the iterated Lagrange multipliers is shown in Fig.~\ref{fig:plot_test_dual_optimization_starlink_constellation_varying_beta_d_o_and_p_o_horizontal}b.
\blue{The converted throughput can also fluctuate across iterations, but it approaches a stable level as the multipliers settle. The network throughput is higher for $\beta=0.5$ than for $\beta=0.7$ and $0.9$, indicating that a slower step-size decay leads to better performance because it explores the multiplier space more thoroughly.}
Also, we observe that the throughput reaches a near-maximum value around $K=500$ iterations. \blue{For the remaining simulations, we fix $\beta=0.5$ and $K=500$.}

\begin{figure}[!t]
  \centering
  \includegraphics[scale=0.85]{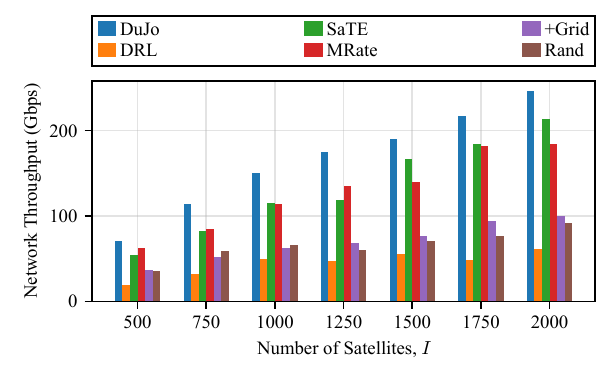}
  \caption{DuJo vs.\ baselines for different numbers of satellites $I$.}
  \label{fig:plot_test_dual_optimization_vs_grid_rand}
\end{figure}

\subsubsection{\blue{Comparison with Baselines}}
\blue{We compare our method with the baseline methods in Fig.~\ref{fig:plot_test_dual_optimization_vs_grid_rand}, where the network throughput is plotted against the number of satellites $I$ in the constellation. The results show that our method improves the network throughput substantially as the constellation size grows. Among the strengthened baselines, SaTE is generally the strongest separated baseline for larger constellations, whereas DRL remains below the heuristic baselines in this varying-$I$ study. This is because the DRL policy is trained using a network-wise reward that is derived from the resulting traffic flow rates, which do not directly supervise how the Lagrange multipliers should be optimized to improve the matching and routing decisions. Thus, the DRL policy may not learn to optimize the Lagrange multipliers effectively, especially when the constellation size grows with many satellite pairs.}
\blue{Compared with +Grid, MRate, and SaTE, our method achieves a higher throughput because it iteratively updates the Lagrange multipliers to improve the LISL matching decisions together with the resulting routing and rate-allocation structure, rather than fixing the topology first.}

\begin{figure}[!t]
  \centering
  \includegraphics[scale=0.85]{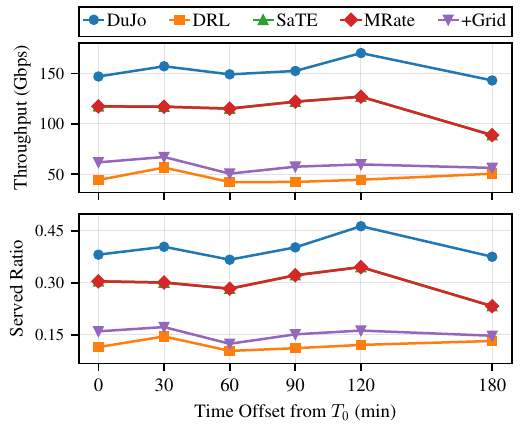}
  \caption{\blue{Structured shell-based Starlink evaluation over time. A coherent $I=1000$ satellite block from the dominant $53.2^\circ$ shell is kept fixed while the constellation is propagated over snapshot offsets of $0$, $30$, $60$, $90$, $120$, and $180$ min. The top plot shows the served throughput and the bottom plot shows the served-demand ratio for DuJo, DRL, SaTE, MRate, and +Grid.}}
  \label{fig:plot_test_tcom_shell_time_baselines}
\end{figure}

\subsubsection{\blue{Structured Shell-Based Evaluation Over Time}}\label{subsubsec:structured_shell_time_evaluation}
\blue{We next evaluate the structured shell-based Starlink instance. Fig.~\ref{fig:plot_test_tcom_shell_time_baselines} shows the results for the same $I=1000$ satellite block over six snapshot times.}
\blue{DuJo achieves the highest served throughput and the highest served-demand ratio at all six snapshot times. SaTE and MRate form the next performance tier, +Grid follows, and DRL remains the lowest among the compared methods.}

\begin{figure}[!t]
  \centering
  \includegraphics[scale=0.85]{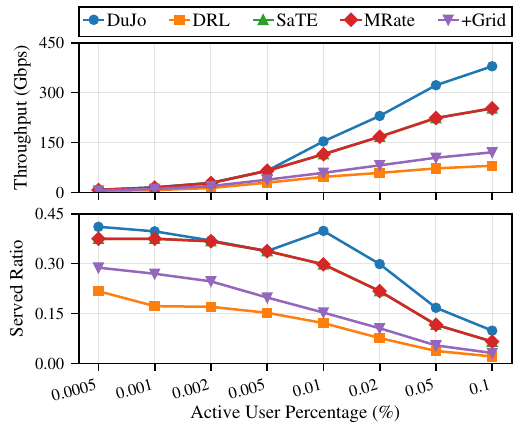}
  \caption{\blue{Load-stress evaluation on the structured shell-based Starlink instance. The active-user percentage varies from $0.0005\%$ to $0.1\%$. The top and the bottom plots show the throughput and the served-demand ratio, respectively.}}
  \label{fig:plot_test_tcom_shell_time_load_scaling}
\end{figure}

\subsubsection{\blue{Load-Stress Evaluation}}\label{subsubsec:load_stress_evaluation}
\blue{For the load-stress evaluation, we vary the active-user percentage while keeping the per-user demand fixed at $D=0.1$ Gbps. Fig.~\ref{fig:plot_test_tcom_shell_time_load_scaling} reports both the served throughput and the served-demand ratio.}
\blue{As the active-user percentage increases, all methods carry more traffic in absolute terms, but the served-demand ratio also decreases, which indicates growing congestion in the gateway and LISL resources. DuJo gives the best reported performance at all tested loads. DRL remains below the heuristic baselines across this load range.
This suggests that the joint optimization remains the most effective approach on this structured-shell benchmark, especially once the network becomes meaningfully loaded.}

\begin{figure}[t]
  \centering
  \includegraphics[scale=0.85]{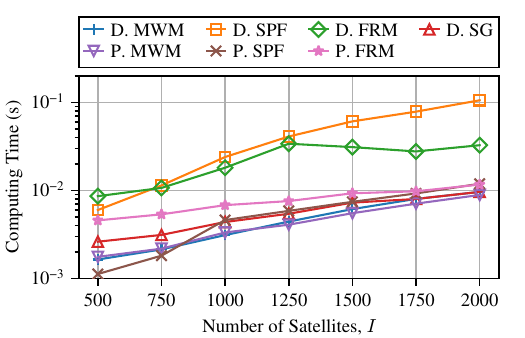}
  \caption{Computing time measurements for the steps in our method, DuJo, with the legends interpreted as follows: the dual-optimization steps are link matching \eqref{eq:routine:mwm_given_lambda}, flow routing \eqref{eq:routine:spf_given_lambda}, rate allocation \eqref{eq:routine:flow_rate_given_routing_cost}, and the subgradient descent step \eqref{eq:routine:dual_variable_update} (with legends ``D. MWM'', ``D. SPF'', ``D. FRM'', and ``D. SG'', respectively); the conversion steps from the Lagrange multipliers are link matching with given Lagrange multipliers \eqref{eq:rounding:mwm}, flow routing with given matched LISLs \eqref{eq:rounding:spf}, and rate allocation with matching and routing decisions \eqref{eq:rounding:frm} (with legends ``P. MWM'', ``P. SPF'', and ``P. FRM'', respectively).
  }
  \label{fig:plot_test_computing_time_measurement}
\end{figure}

\subsubsection{Evaluation of Computational Complexity}
We then evaluate the computational complexity of our method in Fig.~\ref{fig:plot_test_computing_time_measurement} by measuring the computing time of the dual optimization algorithm in Algorithm~\ref{alg:dual_optimization} and the conversion from Lagrange multipliers to the LISL matching and flow routing decisions in \eqref{eq:rounding:mwm}-\eqref{eq:rounding:frm}.
We observe that the computing time increases with the number of satellites $I$ in our method. Here, the conversion of the Lagrange multipliers can be done instantly within $1$ second.
The dual iterations take about $1$ second each. However, the algorithm needs a few hundred iterations to converge, so the total computing time reaches the order of minutes for larger constellations and therefore still requires acceleration.

\subsubsection{\blue{Impact of Constellation Configurations}}
\blue{We vary the LCT configurations on each satellite, and the resulting network throughput is shown in Fig.~\ref{fig:plot_test_dual_optimization_vs_grid_rand_varying_availability_and_for}. The network throughput increases with the average number of LCTs on each satellite because more LCTs allow more LISLs to be established. Similarly, the throughput increases with the FOR because a larger FOR creates more connectable LCT pairs. DuJo consistently outperforms the baseline methods across these sweeps. Among the strengthened baselines, SaTE is generally the strongest separated baseline, while DRL remains weak and does not change the overall ordering. The performance of the grid-based scheme varies little with the FOR because only LCTs near the center of the steering range tend to be connected.}
\blue{Next, we compare the performance of our method in different constellations. DuJo achieves the highest network throughput in all three constellations. In the Starlink case, SaTE is the strongest separated baseline. In the more regular walker-delta and OneWeb constellations, +Grid and SaTE come much closer to DuJo because those regular orbital structures make it easier to establish high-rate LISLs while preserving connectivity for the traffic flows. DRL remains distinctly lower in all three constellations. However, in the less regular Starlink case, the fixed-topology baselines remain poor because they cannot adapt the LISL matching to the routing bottlenecks as effectively as the joint optimization.}

\begin{figure}[t]
  \centering
  \includegraphics[scale=0.85]{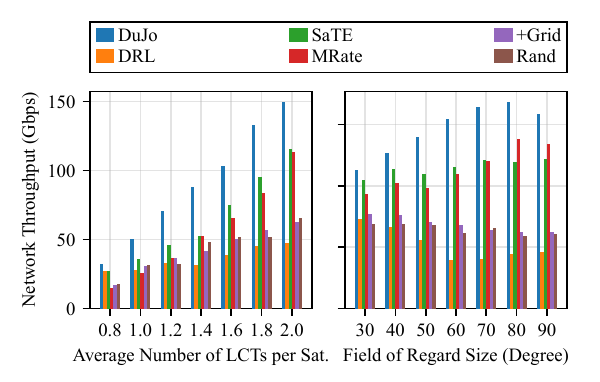}
  \caption{Comparison when varying LCT configurations on satellites.}
  \label{fig:plot_test_dual_optimization_vs_grid_rand_varying_availability_and_for}
\end{figure}

\begin{figure}[t]
  \centering
  \includegraphics[scale=0.85]{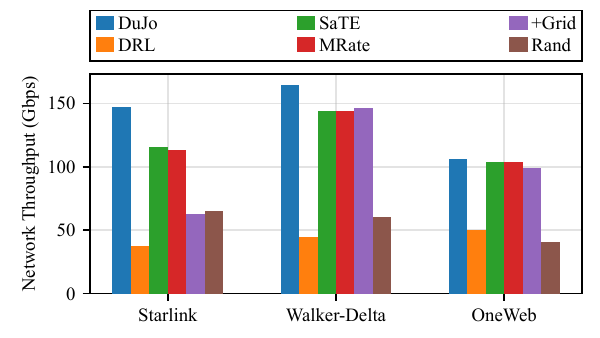}
  \caption{Comparison in a Starlink-based ($I=1000$), a regular walker-delta ($50^{\circ}$ inclination, $I=1000$), and a OneWeb constellation ($I=650$).}
  \label{fig:plot_test_dual_optimization_different_constellation}
\end{figure}

\begin{figure}[!t]
  \centering
  \includegraphics[scale=0.85]{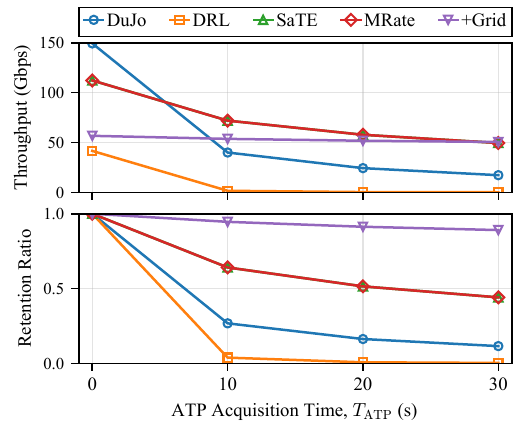}
  \caption{\blue{ATP-time sensitivity evaluation on the structured shell-based Starlink instance. The topology is evaluated every $10$ s over offsets of $0,10,\ldots,360$ s, while the ATP time is swept over $T_{\mathrm{ATP}}\in\{0,10,20,30\}$ s. The top plot shows the average throughput, and the bottom plot shows the retained fraction relative to the original per-snapshot throughput.}}
  \label{fig:plot_test_tcom_atp_transition}
\end{figure}

\subsubsection{\blue{Impact of ATP Time}}\label{subsubsec:atp_transition_evaluation}
\blue{To evaluate non-zero ATP overhead without replacing the proposed per-snapshot optimizer, we perform an ATP-time sensitivity simulation on the same structured shell-based Starlink instance. The topology-evaluation interval is fixed to $10$ s, and the ATP acquisition time is swept from $0$ to $30$ s.}
\blue{For each method, we first compute the original per-snapshot matching, routing, and rate allocation once for each offset. We then replay the ATP transition process for each value of $T_{\mathrm{ATP}}$. A LISL that is decided to be connected at the snapshot $t$ is unavailable over $[t,t+T_{\mathrm{ATP}})$, and it becomes usable once the acquisition time has elapsed. If it is dropped before acquisition completes, its ATP state is discarded. Finally, we recompute the throughput by LP using only the ATP-usable LISLs.}
\blue{The results in Fig.~\ref{fig:plot_test_tcom_atp_transition} show that increasing ATP acquisition time can materially reduce realized throughput when the per-snapshot topology changes quickly. The retention ratios quantify how much of the original per-snapshot throughput remains when ATP overhead is considered. Methods that change many LISLs across nearby snapshots have lower retention ratios as $T_{\mathrm{ATP}}$ increases, whereas methods with more persistent LISLs are less sensitive. These results motivate the future ATP-aware temporal optimization described in the appendix, complementing recent time-dependent topology optimization studies for LEO constellations \cite{ron2025time}.}

\section{Conclusion}\label{sec:conclusion}
This paper presents a Lagrangian duality-based method for joint link matching and traffic routing in satellite mega-constellations. The method applies Lagrangian dual relaxation to the per-link rate constraints and decomposes the joint problem into efficiently solvable subproblems.
Simulations show that our method outperforms the baseline methods in network throughput by iteratively adjusting the Lagrange multipliers to capture both connectivity importance and link congestion, thereby improving LISL matching and flow routing decisions.
\blue{The results also show that the method is practically computable for this mixed-integer formulation, while the associated fixed-snapshot abstract decision problem is NP-hard and the computing time increases for larger constellations.}
Further research could improve the computational efficiency, e.g., by using a proper learning algorithm that integrates the dual optimization into training to predict the Lagrange multipliers directly without iterations \cite{gu2026deepladu}.
The observation on the impact of LCT configurations on the whole constellation suggests optimizing the LCT design to maximize the network profit.
\blue{The ATP-time sensitivity evaluation further shows that non-zero acquisition time can substantially reduce realized throughput when per-snapshot topologies change rapidly. A full ATP-aware temporal optimizer that jointly balances instantaneous throughput and link persistence is therefore an important future direction.}


\section*{Appendix: Proof of Lemma~\ref{lemma:optimal_dual_variable_bounded}}
For given Lagrange multipliers $\lambda$, we map them to new ones as $\lambda'_{i,j} = \min\{1, \lambda_{i,j}\}$, $\forall (i,j)\in\mathcal{L}$.
For $\lambda$, let the set of flows with routing cost greater than or equal to $1$ be $\mathcal{F}^+$ and the set of those with routing cost less than $1$ be $\mathcal{F}^-$. Then consider the flow costs for $\lambda'$. In the set $\mathcal{F}^-$, the routing cost is unchanged since the routing cost was less than $1$ for $\lambda$ and is still the minimum cost for $\lambda'$.
In the set $\mathcal{F}^+$, assuming that a flow's routing decisions change and its cost becomes less than $1$, this new routing cost should also be the minimum cost for $\lambda$ since all links in the path have a cost less than $1$, which contradicts the fact that the minimum routing cost is greater than or equal to $1$ for flows in $\mathcal{F}^+$.
Therefore, the sets $\mathcal{F}^+$ and $\mathcal{F}^-$ remain the same for $\lambda'$ as for $\lambda$.
Since the flow rates in $\mathcal{F}^+$ computed in the dual function are always $0$ (due to the negative weights in the maximization objective) and the routing costs of flows in $\mathcal{F}^-$ remain the same, the value of the routing-cost-weighted flow-rate maximization is the same for $\lambda'$ and $\lambda$, i.e.,
\begin{equation}
  \begin{aligned}
      &\max_{
        \substack{
          \mathbf{q}\in \mathcal{Q}
        }
      }
      \big\{
      \sum_{(s,s')\in\mathcal{F} } q^{s,s'} (1 - \min_{
        \substack{
          \mathbf{x}^{s,s'}\in\mathcal{X}^{s,s'}
        }
      }
      \sum_{(i,j)\in\mathcal{L}}\lambda'_{i,j} x^{s,s'}_{i,j})   
      \big\}\\
      =
      &\max_{
        \substack{
          \mathbf{q}\in \mathcal{Q}
        }
      }
      \big\{
      \sum_{(s,s')\in\mathcal{F} } q^{s,s'} (1 - \min_{
        \substack{
          \mathbf{x}^{s,s'}\in\mathcal{X}^{s,s'}
        }
      }
      \sum_{(i,j)\in\mathcal{L}}\lambda_{i,j} x^{s,s'}_{i,j})   
      \big\}.
\end{aligned}
\notag
\end{equation}
Also, the MWM with $\lambda'$ yields matched weights no greater than those with $\lambda$ since $\lambda'\leq\lambda$, as
\begin{equation}
  \begin{aligned}
&\max_{
  \substack{
    \mathbf{c}\in\mathcal{C}
  }
}
\big\{
\sum_{\{n,m\}\in\mathcal{E}} (\lambda'_{i_n,i_m} + \lambda'_{i_m,i_n})  r_{n,m} c_{n,m}
\big\}\\ 
\leq&
\max_{
  \substack{
    \mathbf{c}\in\mathcal{C}
  }
}
\big\{
\sum_{\{n,m\}\in\mathcal{E}} (\lambda_{i_n,i_m} + \lambda_{i_m,i_n})  r_{n,m} c_{n,m}
\big\}.
  \end{aligned}
  \notag
\end{equation}
Due to the above facts, $g(\lambda')\geq g(\lambda)$. Since we can always find such $\lambda'$ for any $\lambda$ as well as for any $\lambda^*$, there exist optimal Lagrange multipliers $\lambda^*$ with all elements less than or equal to $1$. Thus, the difference between $\lambda^{[1]}=\mathbf{0}$ and $\lambda^*$ can be bounded by $\|\mathbf{1}^{|\mathcal{E}|\times1}\|$. The boundedness of the maximum norm of the supergradient $\max_\lambda \|\delta(\lambda)\|$ is due to the fact that the traffic flow rates in $\mathcal{Q}$ are bounded and the capacities of LISLs are bounded by the maximum at zero distance.

\section*{Appendix: Proof of Theorem 1}
The Lagrange multipliers in the $k$-th iteration are denoted by $\lambda^{[k]}$. The optimal Lagrange multipliers are denoted by $\lambda^*$.
Denote the set of all possible supergradients of $g(\lambda^{[k]})$ at $\lambda^{[k]}$ as $\hat{\partial}_{\lambda^{[k]}} g(\lambda^{[k]}) \triangleq \{\delta| g(\lambda')\leq g(\lambda^{[k]}) + \delta^{\rm T}(\lambda'-\lambda^{[k]}), \forall \lambda'\}$.
$\delta(\lambda^{[k]})$ is a supergradient of $g(\lambda^{[k]})$ at $\lambda^{[k]}$, i.e., $ \delta(\lambda^{[k]}) \in \hat{\partial}_{\lambda^{[k]}}  g(\lambda^{[k]})$.
The difference between the optimal Lagrange multipliers $\lambda^*$ and the iterated ones $\lambda^{[k+1]}$ is given by
\begin{equation}
  \begin{aligned}
    &\|\lambda^{[k+1]} - \lambda^*\|^2 = \|\lambda^{[k]} + \alpha^{[k]} \delta(\lambda^{[k]}) - \lambda^*\|^2\\
  = &\|\lambda^{[k]} - \lambda^*\|^2 + 2\alpha^{[k]}\delta(\lambda^{[k]})^{\rm T} (\lambda^{[k]} - \lambda^*) + \|\alpha^{[k]}\delta(\lambda^{[k]})\|^2\\
\stackrel{\text{(a)}}{\leq}
    &\|\lambda^{[k]} - \lambda^*\|^2 + 2\alpha^{[k]} (g(\lambda^{[k]}) - g(\lambda^*)) + (\alpha^{[k]})^2 \|\delta(\lambda^{[k]})\|^2,
  \end{aligned}
  \notag
\end{equation}
where (a) is due to the definition of the supergradient. By applying telescoping sum, we have
\begin{equation}
  \begin{aligned}
    &\|\lambda^{[k+1]} - \lambda^*\|^2 \leq \|\lambda^{[1]} - \lambda^*\|^2 \\
    &+ 2 \sum_{i=1}^{k}\alpha^{[i]} (g(\lambda^{[i]}) - g(\lambda^*)) + \sum_{i=1}^{k}(\alpha^{[i]})^2 \|\delta(\lambda^{[i]})\|^2.
  \end{aligned}
  \notag
\end{equation}
Note that $g(\lambda^{[i]})\leq g(\hat{\lambda}^{[k]})$ $\forall i$ based on the definition of $\hat{\lambda}^{[k]}$, implying $\sum_{i=1}^{k}\alpha^{[i]} (g(\lambda^{[i]}) - g(\lambda^*)) \leq  \big(g(\hat{\lambda}^{[k]}) - g(\lambda^*)\big)\sum_{i=1}^{k}\alpha^{[i]}$.
Also, $\|\lambda^{[k+1]} - \lambda^*\|^2$ is non-negative. Substituting these facts in the above inequality, we have
\begin{equation}
  \begin{aligned}
    &g(\lambda^*) - g(\hat{\lambda}^{[k]}) \leq \frac{\|\lambda^{[1]} - \lambda^*\|^2+\sum_{i=1}^{k}(\alpha^{[i]})^2 \|\delta(\lambda^{[i]})\|^2}{2\sum_{i=1}^{k}\alpha^{[i]}}.
  \end{aligned}
  \notag
\end{equation}
Using the convergence of $\sum_{i=1}^{k}\alpha^{[i]}$ and $\sum_{i=1}^{k}(\alpha^{[i]})^2$ when $0.5\leq \beta < 1$ \cite{boyd2003subgradient} and the boundedness of $\|\lambda^{[1]} - \lambda^*\|^2$ and $\max_\lambda \|\delta(\lambda)\| $ from Lemma~\ref{lemma:optimal_dual_variable_bounded}, we prove the statement.

{\color{blue}
\section*{Appendix: Proof of Proposition \ref{prop:np_hardness}}
\begin{proof}
We reduce from the integral 2-commodity flow problem on a symmetric digraph $\mathcal{G}$, which is NP-complete \cite[Theorem 1]{jarry2012multiflows}. Each unordered edge $\{u,v\}$ in $\mathcal{G}$ has an integer capacity $\kappa(u,v)=\kappa(v,u)$ that applies to both directed arcs $(u,v)$ and $(v,u)$. The instance has two requests $(s_1,t_1,d_1)$ and $(s_2,t_2,d_2)$, where $s_1,s_2,t_1,t_2$ are distinct vertices in $\mathcal{G}$, and $d_1$ and $d_2$ are positive integers denoting the required flow values. Set $D=d_1+d_2$. The symmetric source problem matches AS-JMR because an established LISL is available in both directions by \eqref{eq:const:link:binary_bidirectional_connection}.

We construct an AS-JMR instance as follows. First, form the AS-JMR satellite set by creating one satellite for each vertex in $\mathcal{G}$. For each commodity $h\in\{1,2\}$ and each unit $p=1,\dots,d_h$, add two new terminal satellites to this AS-JMR satellite set, a source-terminal satellite $a_{h,p}$ and a target-terminal satellite $b_{h,p}$, and include $(a_{h,p},b_{h,p})$ in $\mathcal{F}$. Set $Q_{a_{h,p}}=1$ and $D_{b_{h,p}}=1$ for these terminal pairs, set all other serving and demand bounds to zero, and set the AS-JMR threshold to $\zeta=D$. For each unordered source edge $\{u,v\}$, create one dedicated connectable LCT pair between satellites $u$ and $v$ with rate $\kappa(u,v)$, and create no other real connectable pair between them. For each terminal pair $(a_{h,p},b_{h,p})$, create one unit-rate attachment LCT pair between $a_{h,p}$ and $s_h$ and one unit-rate attachment LCT pair between $t_h$ and $b_{h,p}$. Finally, add isolated LCTs as needed so that all satellites have the same number of LCTs; these added LCTs are not included in any connectable pair and therefore do not affect routing or matching feasibility.

If the source instance is feasible, decompose each integral commodity flow into unit directed paths. For every unit path of commodity $h$, route the corresponding AS-JMR pair $(a_{h,p},b_{h,p})$ through the attachment edge $a_{h,p}\to s_h$, the selected unit $s_h$--$t_h$ path in $\mathcal{G}$, and the attachment edge $t_h\to b_{h,p}$. For every unordered source edge $\{u,v\}$ used by at least one of these paths, establish the dedicated LCT pair created between satellites $u$ and $v$; this bidirectional LISL provides rate $\kappa(u,v)$ on each directed arc $(u,v)$ and $(v,u)$. We also establish all attachment LCT pairs. Since the source flow uses each directed arc $(u,v)$ at most $\kappa(u,v)$ times, the AS-JMR link-rate constraints hold. Thus all $D$ unit terminal pairs are served and the AS-JMR throughput is $D=\zeta$.

Conversely, suppose the constructed AS-JMR instance has a feasible solution with throughput at least $D$. There are exactly $D$ terminal-pair variables $q^{a_{h,p},b_{h,p}}$, each upper bounded by one, so every one of them must equal one. For each such pair, the binary routing constraints define an integral unit flow from $a_{h,p}$ to $b_{h,p}$. After deleting cycles, the path must begin with $a_{h,p}\to s_h$ because $a_{h,p}$ has no other real neighbor, and must end with $t_h\to b_{h,p}$ because $b_{h,p}$ has no other real neighbor. Removing these two attachment arcs gives a directed $s_h$--$t_h$ path in the original graph. Aggregating these extracted paths for each commodity gives integral flows of values $d_1$ and $d_2$. The capacity constraints are respected because the only real connectable LCT pair between any adjacent original satellites $u$ and $v$ has rate $\kappa(u,v)$. Hence the source instance is feasible. The construction is polynomial in the source instance size because it creates one satellite per source vertex, $D$ source-terminal satellites, $D$ target-terminal satellites, and one dedicated LCT pair per unordered source edge; the isolated-LCT padding is polynomial in the number of created real LCTs. Therefore, AS-JMR is NP-hard.
\end{proof}

\section*{Appendix: An MDP Extension for Future Work}
This appendix outlines an example discrete-time MDP formulation in which the Lagrangian decomposition of Section~\ref{sec:lagrangian_dual_relaxation} carries over at each snapshot and the entire temporal coupling is absorbed into the matching subproblem. The formulation is presented as a future-work direction, in line with recent time-dependent topology optimization for LEO constellations \cite{ron2025time}. Exact solution algorithms, convergence analysis, and detailed ATP-time models are outside the scope of this appendix; the ATP-time model below is an illustrative abstraction \cite{kaymak2018survey,li2011analytical,scheinfeild2000acquisition,bashir2021adaptive}.

\noindent\textit{MDP Formulation.} We discretize the continuous time $t$ of Section~\ref{sec:system_model} into snapshots of duration $\Delta t$, indexed by $\ell=0,1,\dots$ with snapshot time $t_\ell=\ell\Delta t$, and attach subscript $\ell$ to snapshot-indexed quantities (e.g., $\mathbf{c}_\ell,\mathbf{x}_\ell,\mathbf{q}_\ell$). For each $\{n,m\}\in\mathcal{E}$, let $\tau_{n,m}\in\mathbb{Z}_{\geq 0}$ denote the ATP time (in snapshots) required after a link is newly established, and let $h_{\ell,n,m}\in\{0,\dots,\tau_{n,m}\}$ be its remaining countdown; the link is usable iff $h_{\ell,n,m}=0$. Let $\omega_\ell=(\{\mathbf{l}_i(t_\ell)\}_i,\{U_i(t_\ell)\}_i)$ collect deterministic geometry and stochastic traffic. The state, action, and ATP transition are
\begin{equation}
\begin{aligned}
&\xi_\ell = (\mathbf{c}_{\ell-1},\mathbf{h}_\ell,\omega_\ell),\ \ u_\ell=(\mathbf{c}_\ell,\mathbf{x}_\ell,\mathbf{q}_\ell)\in\mathcal{C}\!\times\!\mathcal{X}\!\times\!\mathcal{Q},\\
&h_{\ell+1,n,m}=\begin{cases}\tau_{n,m}&\!\!c_{\ell,n,m}\!=\!1,c_{\ell-1,n,m}\!=\!0,\\ (h_{\ell,n,m}\!-\!1)_{+}&\!\!c_{\ell,n,m}\!=\!1,c_{\ell-1,n,m}\!=\!1,\\ 0&\!\!c_{\ell,n,m}=0.\end{cases}
\end{aligned}\notag
\end{equation}
Define the ATP-modulated rate $r^{\text{eff}}_{\ell,n,m}=\mathbb{I}\{h_{\ell,n,m}=0\}r_{n,m}(t_\ell)$, and accrue the per-snapshot throughput $R_\ell=\sum_{(s,s')}q^{s,s'}_\ell$ subject to \eqref{eq:const:path:flow_rate} with $r_{n,m}$ replaced by $r^{\text{eff}}_{\ell,n,m}$. The objective is to maximize the long-term expected discounted throughput $\mathbb{E}^{\pi}\big[\sum_{\ell}\gamma^\ell R_\ell\big]$ for discount $\gamma\in(0,1)$.

\noindent\textit{Per-snapshot Lagrangian decomposition.} Introducing multipliers $\lambda_\ell\geq 0$ and lifting the link-rate constraint, the Bellman equation of the throughput-maximization MDP takes the form
\begin{equation}
\tilde{V}(\xi_\ell)=\min_{\lambda_\ell\geq 0}\max_{u_\ell}\Big\{\tilde{L}_\ell(u_\ell,\lambda_\ell)+\gamma\mathbb{E}[\tilde{V}(\xi_{\ell+1})|\xi_\ell,\mathbf{c}_\ell]\Big\},\notag
\end{equation}
where $\tilde{L}_\ell$ is the throughput-maximization form of the per-snapshot Lagrangian (the negative of \eqref{eq:prob:constrainted_routing_and_matching:lagrangian} with $r^{\text{eff}}_{\ell,n,m}$ in place of $r_{n,m}$). Since the continuation value depends on $u_\ell$ only through the matching variable $\mathbf{c}_\ell$, the inner maximization separates exactly as in \eqref{eq:prob:constrainted_routing_and_matching:lagrangian_dual}: subproblems (b) shortest-path routing and (c) LP rate allocation remain per-snapshot and memoryless, while subproblem (a) becomes
\begin{equation}
\begin{aligned}
&\hat{\mathbf{c}}_\ell(\xi_\ell,\lambda_\ell)=\arg\max_{\mathbf{c}\in\mathcal{C}}\Big\{\gamma\mathbb{E}[\tilde{V}(\xi_{\ell+1})|\xi_\ell,\mathbf{c}]\\
&\qquad+\!\!\sum_{\{n,m\}\in\mathcal{E}}\!\!(\lambda_{\ell,i_n,i_m}\!+\!\lambda_{\ell,i_m,i_n})r^{\text{eff}}_{\ell,n,m}c_{n,m}\Big\},
\end{aligned}\notag
\end{equation}
i.e., a one-step-lookahead MDP-matching with Lagrangian-weighted per-step rewards and continuation value $\tilde{V}$. Candidate approximation methods for $\tilde{V}$ include approximate dynamic programming, deep reinforcement learning, and graph representation learning over the ISL topology \cite{gu2021knowledge,gu2024graph}. The multipliers $\lambda_\ell$ can be updated per snapshot or maintained as a stationary dual policy; developing and analyzing these solvers remains future work.

\noindent\textit{Degenerate case.} When $\tau_{n,m}=0$ for all pairs, $r^{\text{eff}}_{\ell,n,m}\equiv r_{n,m}(t_\ell)$, the matching decouples from the past, and the continuation term in (a) reduces to a constant, so subproblem (a) collapses to the per-snapshot MWM \eqref{eq:routine:mwm_given_lambda} and the proposed per-snapshot formulation is recovered as the ATP-free case of this MDP extension.}

\bibliography{main}
\bibliographystyle{IEEEtran}

\end{document}